\documentclass[11pt,letterpaper,oneside,english]{article}
\usepackage[left=3.5cm,right=3.5cm,top=3cm,bottom=3cm]{geometry}

\usepackage{mdwlist,enumerate}
\usepackage{datetime}

\makeatletter
\let\@font@warningori\@font@warning
\newcommand\shutup{\def\@font@warning##1{}}
\newcommand\youcanspeak{\let\@font@warning\@font@warningori}
\makeatother 
\usepackage{amsmath,amsbsy,bm,latexsym,amssymb}
\shutup
\usepackage{fourier}
\youcanspeak
\usepackage[scaled=0.875]{helvet}

\usepackage{graphics, subfigure, float}
\usepackage{fp, calc}

\usepackage[latin1]{inputenc}
\usepackage[american]{babel}
\usepackage[T1]{fontenc} 

\usepackage{amscd,amsthm}

\usepackage{verbatim, comment}

\usepackage[dvips]{graphicx,epsfig,color}
\usepackage{pst-all}
\usepackage{pstricks-add}

\theoremstyle{theorem}
\newtheorem{theorem}{Theorem}
\newtheorem*{theorem*}{Theorem}

\newtheorem{lemma}[theorem]{Lemma}
\newtheorem*{lemma*}{Lemma}
\newtheorem{claim}[theorem]{Claim}
\newtheorem*{claim*}{Claim}

\newtheorem*{conjecture*}{Conjecture}
\newtheorem*{problem*}{Problem}

\newtheorem*{definition*}{Definition}
\newtheorem{observation}{Observation}

\theoremstyle{remark}

\newtheorem*{remark*}{Remark}

\newtheorem*{algorithm*}{Algorithm}

\providecommand{\setN}{\mathbb{N}}
\providecommand{\setZ}{\mathbb{Z}}

\providecommand{\setR}{\mathbb{R}}

\usepackage{todonotes}
        \def\drawRect#1#2#3#4#5{
           \FPeval{\x2}{(#2) + (#4)} 
           \FPeval{\y2}{(#3) + (#5)} 
           \pspolygon[#1](#2,#3)(\x2,#3)(\x2,\y2)(#2,\y2)
        }

\psset{arrowsize=6pt, labelsep=2pt, linewidth=1.5pt}
\makeatother

\usepackage{calrsfs} 
\DeclareMathAlphabet{\pazocal}{OMS}{zplm}{m}{n}

\title{A Logarithmic Additive Integrality Gap for Bin Packing}

\author{Rebecca Hoberg\thanks{Email: {\tt rahoberg@math.washington.edu}} \quad and \quad Thomas Rothvo{ss}\thanks{Email: {\tt rothvoss@uw.edu}. Both authors supported by NSF grant 1420180 with title ``\emph{Limitations of convex relaxations in combinatorial optimization}''. File compiled on {\today, \currenttime}.} 
\vspace{2mm} \\ University of Washington, Seattle} 

\begin{document}

\maketitle

\begin{abstract}
\noindent 
For \emph{bin packing}, the input consists of $n$ items 
with sizes $s_1,\ldots,s_n \in [0,1]$ which have to be assigned to a minimum
number of bins of size 1. Recently, the second author gave an LP-based polynomial time 
algorithm that employed techniques from \emph{discrepancy theory} to find a solution using at most $OPT + O(\log OPT \cdot \log \log OPT)$ bins.

In this paper, we present an approximation algorithm that has an additive gap of
only $O(\log OPT)$ bins, which matches certain combinatorial lower bounds. Any
further improvement would have to use more algebraic structure. 
Our improvement is based on a combination of discrepancy theory techniques and a novel 
\emph{2-stage packing}: first we pack items into \emph{containers}; 
then we pack containers into bins of size 1. Apart from being more effective, we
believe our algorithm is much cleaner than the one of Rothvoss.
\end{abstract}

\section{Introduction}

One of the classical combinatorial optimization problems that is studied in 
computer science is \emph{Bin Packing}. 
It appeared as one of the prototypical $\mathbf{NP}$-hard problems already in 
the book of Garey and Johnson~\cite{GareyJohnson79} but it was studied long before
in operations research in the 1950's, for example by~\cite{TrimProblem-Eiseman1957}.
We refer to the survey of Johnson~\cite{BinPackingSurvey84} for a complete historic account. Bin packing is a good example to study the development of techniques in approximation
algorithms as well. The 1970's brought simple greedy heuristics such as \emph{First Fit}, analyzed by
Johnson~\cite{Johnson73} which requires  at most $1.7\cdot OPT + 1$ bins and \emph{First Fit Decreasing}~\cite{JohnsonFFD74}, which yields a solution with 
at most $\frac{11}{9} OPT + 4$
bins (see~\cite{FFDtightBound-Dosa07} for a tight bound of $\frac{11}{9} OPT + \frac{6}{9}$). 
Later, an \emph{asymptotic PTAS} was developed by Fernandez de la Vega and Luecker~\cite{deLaVegaLueker81}. One of their main technical contributions was an \emph{item grouping technique} to reduce the number of different item types. 
The algorithm of De la Vega and Luecker finds solutions using at most $(1+\varepsilon)OPT + O(\frac{1}{\varepsilon^2})$ bins, while the running time is either of the form $O(n^{f(\varepsilon)})$ if one uses
dynamic programming or of the form $O(n \cdot f(\varepsilon))$ if one applies 
linear programming techniques. 

A big leap forward in approximating bin packing was done by Karmarkar and Karp
in 1982~\cite{KarmarkarKarp82}. First of all, they argue how a certain exponential size
LP can be approximately solved in polynomial time; secondly they provide a
sophisticated rounding scheme which produces a solution with at most
$OPT + O(\log^2 OPT)$ bins, corresponding to an
\emph{asymptotic FPTAS}.

It will be convenient throughout this paper to allow a more compact form of input, where
$s \in [0,1]^n$ denotes the vector of different item sizes and $b \in \setN^n$ denotes the
multiplicity vector, meaning that we have $b_i$ copies of item type $i$. In this notation
we say that $\sum_{i=1}^n b_i$ is the \emph{total number of items}. 
The linear program that we mentioned earlier 
is called the \emph{Gilmore-Gomory LP relaxation}~\cite{TrimProblem-Eiseman1957,Gilmore-Gomory61} and it is of the form 
\begin{equation} \label{eq:GilmoreGomory}
 \min\left\{ {\bf{1}}^Tx \mid Ax \geq b, x \geq \bm{0} \right\}.
\end{equation}
Here, the constraint matrix $A$ consists of all column vectors $p \in \setZ_{\geq 0}^n$ 
that satisfy $\sum_{i=1}^n p_i s_i \leq 1$. The linear program has 
variables $x_p$ that give the number of bins that should be packed according to the
\emph{pattern} $p$.

We denote the value of the optimal fractional solution to \eqref{eq:GilmoreGomory}
by $OPT_f$, and the value of the best integral solution by $OPT$.
As we mentioned before, the linear program~\eqref{eq:GilmoreGomory} does have an 
exponential number of variables, but only $n$ constraints. A fractional solution $x$ 
of cost $\bm{1}^Tx \leq OPT_f + \delta$ can be computed in time polynomial in $\sum_{i=1}^n b_i$ and
$1/\delta$~\cite{KarmarkarKarp82} using the Gr{\"o}tschel-Lovasz-Schrijver
variant of the Ellipsoid method~\cite{GLS-algorithm-Journal81}. 
An alternative and simpler way to solve the LP approximately is via
the Plotkin-Shmoys-Tardos framework~\cite{FractionalPackingAndCovering-PlotkinShmoysTardos-Journal95} or the
multiplicative weight update method. See the survey of \cite{MWU-Survey-Arora-HazanKale2012} for an overview.

The best known lower bound on the integrality gap of the Gilmore-Gomory LP is
an instance where $OPT = \left\lceil  OPT_f \right\rceil + 1$; Scheithauer and Terno~\cite{BinPacking-MIRUP-ScheithauerTerno97}
conjecture that these instances represent the worst case additive gap. 
While this conjecture is still open, it is understandable that 
the best approximation algorithms are based on rounding a solution
to this amazingly strong Gilmore Gomory LP relaxation. 
For example, the Karmarkar-Karp algorithm operates in $\log n$ iterations in which one 
first groups the items such that only $\frac{1}{2}\sum_{i \in [n]} s_i$ many different item sizes 
remain; then one computes a basic solution $x$ and buys $\lfloor x_p\rfloor$ times 
pattern $p$ and continues with the residual instance. 
The analysis provides a $O(\log^2 OPT)$ upper bound on the 
\emph{additive} integrality gap of \eqref{eq:GilmoreGomory}.

The rounding mechanism in the recent paper of the second author~\cite{DBLP:conf/focs/Rothvoss13} uses an algorithm by 
Lovett and Meka that was originally designed for discrepancy minimization. 
The Lovett-Meka algorithm~\cite{DiscrepancyMinimization-LovettMekaFOCS12} can be conveniently summarized as
follows: 
\begin{theorem}[Lovett-Meka '12]
Let $v_1,\ldots,v_m \in \setR^n$ be vectors with $x_{\textrm{start}} \in [0,1]^n$ and parameters
$\lambda_1,\ldots,\lambda_m \geq 0$ so that $\sum_{j=1}^m e^{-\lambda_j^2/16} \leq \frac{n}{16}$. Then in 
randomized polynomial time one can find a vector $x_{\textrm{end}} \in [0,1]^n$ so that $\left|\left<x_{\textrm{end}}-x_{\textrm{start}},v_j\right>\right| \leq \lambda_j \cdot \|v_j\|_2$ for
all $j \in \{ 1,\ldots,m\}$
and at least half of the entries of $x_{\textrm{end}}$ are in $0/1$.
\end{theorem}
Intuitively, the points $x_{\textrm{end}}$ satisfying the linear constraints 
$|\left<x_{\textrm{end}} - x_{\textrm{start}}\right>| \leq \lambda_j \cdot \|v_j\|_2$
form a polytope and the distance of the $j$th hyperplane to the start point is exactly $\lambda_j$. 
Then the condition $\sum_{j=1}^m e^{-\lambda_j^2/16} \leq \frac{n}{16}$ essentially says that the 
polytope is going to be ``large enough''. The algorithm of \cite{DiscrepancyMinimization-LovettMekaFOCS12} 
itself consists of a random walk through the polytope. For more details, we refer to the very readable paper 
of \cite{DiscrepancyMinimization-LovettMekaFOCS12}. 
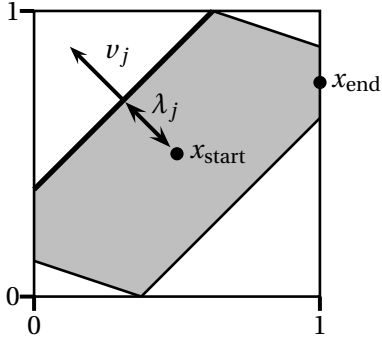
\begin{center}
 \psset{unit=1.9cm}
 \begin{pspicture}(-1.0,-1.2)(1.1,1.2)
\pspolygon[linestyle=none,fillstyle=solid,fillcolor=lightgray](-1,-0.25)(0.25,1)(1,0.75)(1,0.25)(-0.25,-1)(-1,-0.75)
 \pspolygon[linewidth=1pt](-1,-1)(-1,1)(1,1)(1,-1)(-1,-1)
\cnode*(0,0){2.5pt}{origin} \nput{0}{origin}{$x_{\textrm{start}}$}
 \psline[linewidth=1pt](-1,-0.25)(0.25,1)
 \psline[linewidth=1pt](1,0.25)(-0.25,-1)
 \psline[linewidth=1pt](0.25,1)(1,0.75)
 \psline[linewidth=1pt](-0.25,-1)(-1,-0.75)
 \pnode(-0.375,0.375){A} \pnode(-0.75,0.75){B} 
\psline[linewidth=2pt](-1,-0.25)(0.25,1) \ncline[arrowsize=5pt]{->}{A}{B} \nbput[labelsep=2pt]{$v_j$} 
 \pnode(1,0.5){y}
   \psdots[linecolor=black,linewidth=1.5pt](y)
   \nput[labelsep=4pt]{0}{y}{$x_{\textrm{end}}$}
 \psline(-1,-1)(-1,-1.1) \rput[c](-1,-1.2){$0$} 
 \psline( 1,-1)( 1,-1.1) \rput[c](1,-1.2){$1$}
 \psline(-1,-1)(-1.1,-1) \rput[r](-1.1,-1){$0$} 
 \psline(-1, 1)(-1.1, 1) \rput[r](-1.1, 1){$1$}
\ncline[linecolor=black,linewidth=1.5pt,arrowsize=6pt,nodesepA=1pt,nodesepB=1pt]{<->}{origin}{A} \nbput[labelsep=0pt]{$\lambda_j$} 
 \end{pspicture}
\end{center}

The bin packing approximation algorithm of Rothvoss~\cite{DBLP:conf/focs/Rothvoss13} 
consists of logarithmically many runs of Lovett-Meka. 
To be able to use the Lovett-Meka algorithm effectively, Rothvoss needs to rebuild the instance
in each iteration and ``glue'' clusters of small items together to larger items. His procedure is only 
able to do that for items that have size at most 
$\frac{1}{\textrm{polylog}(n)}$ and each of the iterations incurs a loss in the objective function of $O(\log \log n)$. 
In contrast we present a procedure that can even cluster items together that have size 
up to $\Omega(1)$. Moreover, Rothvoss' algorithm only uses two types of parameters
for the error parameters, namely $\lambda_j \in \{ 0,O(\sqrt{\log \log n})\}$. In contrast, we use the
full spectrum of parameters to achieve only a \emph{constant} loss in each of the
logarithmically many iterations.

\subsection{Our contribution}

Our main contribution is the following theorem:
\begin{theorem} \label{thm:MainContribution}
For any Bin Packing instance $(s,b)$ with $s_1,\ldots,s_n \in [0,1]$, one can compute 
a solution with at most $OPT_f + O(\log OPT_f)$ bins, 
 where $OPT_f$ denotes the optimal value of the Gilmore-Gomory LP relaxation. 
The algorithm is randomized and the expected running time is polynomial in $\sum_{i=1}^n b_i$.
\end{theorem}
The recent book of Williamson and Shmoys~\cite{DesignOfApproxAlgosWilliamsonShmoys} 
 presents a list of 10 open problems in approximation algorithms. 
Problem $\#3$ in the list is whether the Gilmore-Gomory LP has a constant integrality gap; hence we make 
 progress towards that question.  

We want to remark that the original algorithm of Karmarkar and Karp has an additive approximation ratio of
 $O(\log OPT_f \cdot \log( \max_{i,j} \{ \frac{s_i}{s_j} \}))$. For \emph{3-partition} instances where all item sizes
are strictly between $\frac{1}{4}$ and $\frac{1}{2}$, this results in an $O(\log n)$ guarantee, which coincides
with the guarantees of Rothvoss~\cite{DBLP:conf/focs/Rothvoss13} and this paper if applied to those instances. 
A paper of Eisenbrand et al.~\cite{BinPackingViaPermutationsSODA2011} gives a reduction of those instances
to minimizing the discrepancy of 3 permutations. Interestingly, shortly afterwards 
 Newman and Nikolov~\cite{CounterexampleBecksPermutationConjecture-FOCS12} showed that there are
instances of 3 permutations that do require a discrepancy of $\Omega(\log n)$. It seems unclear how to
realize those permutations with concrete sizes in a bin packing instance --- however any further improvement for
bin packing even in that special case with item sizes in $]\frac{1}{4},\frac{1}{2}[$ would need to 
rule out such a realization as well. The second author is willing to conjecture that the integrality gap
for the Gilmore Gomory LP is indeed $\Theta(\log n)$.

\section{A 2-stage packing mechanism} \label{sec:Deficiency}

It is well-known that for the kind of approximation guarantee that we aim to achieve, 
one can assume that the items are not too tiny. 
In fact it suffices to prove
an additive gap of $O(\log \max\{ n,\frac{1}{s_{\min}} \})$ where $n$ is the number 
of different item sizes and $s_{\min}$ is a lower bound on all item sizes. 
Note that in the following,  ``polynomial time'' means 
always polynomial in the total number of items $\sum_{i=1}^n b_i$. 
\begin{lemma} \label{lem:AssumptionSizesAtLeast1-n}
Assume for a monotone function $f$, there is a polynomial time  $OPT_f + f(\max\{n,\frac{1}{s_{\min}}\})$ algorithm 
for Bin Packing instances $(s,b)$ with $s \in [0,1]^n$ and $s_1,\ldots,s_{n} \geq s_{\min} > 0$. 
Then there is a polynomial time algorithm that for \emph{all} instances finds a solution 
with at most 
$OPT_f + f(OPT_f) + O(\log OPT_f)$ bins.
\end{lemma}
For a proof, we refer to Appendix~A. From now on we assume that we have $n$ different item sizes with all sizes
satisfying $s_i \geq s_{\min}$ for some given parameter $s_{\min}$ (as a side remark, the reduction in Lemma~\ref{lem:AssumptionSizesAtLeast1-n} will  choose $s_{\min} = \Theta(\frac{1}{OPT_f})$). 
Starting from a fractional solution $x$ to \eqref{eq:GilmoreGomory}
our goal is to find an integral solution of cost $\bm{1}^Tx + O(\log \max\{n,\frac{1}{s_{\min}}\})$. Another useful standard argument is
as follows:
\begin{lemma} \label{lem:GreedyPacking}
Any bin packing instance $(s,b)$ can be packed in polynomial time into at most $2\sum_{i=1}^n s_ib_i+1$ bins.
\end{lemma}
\begin{proof}
Simply assign the items greedily and open new bins only if necessary. If we end up with $k$ bins, then
at least $k-1$ of them are at least half full, which means that $\sum_{i=1}^n s_ib_i \geq \frac{1}{2}(k-1)$. 
Rearranging gives the claim.
\end{proof}
Now, we come to the main mechanism that allows us the improvement over Rothvoss~\cite{DBLP:conf/focs/Rothvoss13}. Consider an instance $(s,b)$ and a fractional LP solution $x$. We could imagine 
the assignment of items in the input to slots in $x$ as a fractional matching in a bipartite
graph, where we have nodes $i \in [n]$ on the left hand side, each with \emph{demand} $b_i$ and
nodes $(p,i)$ on the right hand side with \emph{supply} $x_p \cdot p_i$. 
Instead, our idea is to employ a \emph{2-stage packing}: first we pack items into \emph{containers}, 
then we pack containers into bins. Here, a container is a multiset of items. 
Before we give the formal definition, we want to explain our construction with a small example
that is visualized in Figure~\ref{fig:AssignmentExample}.
The example has $n=3$ items of size $s = (0.3,0.2,0.1)$ and multiplicity vector $b = (2,1,7)$. 
Those items are assigned into
containers $C_1,C_2,C_3$ which also have multiplicities. In this case we have $y_{C_1} = y_{C_2} = 1$
copies of the first two containers and $y_{C_3} = 2$ copies of the third container. 
Moreover, in our example we have 3 patterns $p_1,p_2,p_3$ each with fractional value $x_{p_1}=x_{p_2} = x_{p_3} = \frac{1}{2}$.
For example, item $2$ is packed into container $C_1$ and that container is assigned with a fractional value of
$\frac{1}{2}$ each to pattern $p_2$ and $p_3$.
\begin{figure}
\begin{center}
\psset{xunit=4cm,yunit=0.8cm}
\begin{pspicture}(0,-0.5)(3,9)
\rput[r](-0.2,8.5){items}
\rput[r](-0.2,4.5){containers}
\rput[r](-0.2,0.5){bins}
\drawRect{fillcolor=black!50!white,fillstyle=solid,linewidth=0.5pt}{0.4}{4}{0.6}{1}
\drawRect{fillcolor=black!40!white,fillstyle=solid,linewidth=0.5pt}{1.6}{4}{0.4}{1}
\drawRect{fillcolor=black!20!white,fillstyle=solid,linewidth=0.5pt}{2.6}{4}{0.3}{1}
\rput[c](0.4,3.7){\psline{|-|}(0,0)(0.6,0) \cnode*(0.3,0){3.0pt}{C1}}
\rput[c](1.6,3.7){\psline{|-|}(0,0)(0.4,0) \cnode*(0.2,0){2.5pt}{C2}}
\rput[c](2.6,3.7){\psline{|-|}(0,0)(0.3,0) \cnode*(0.15,0){2.5pt}{C3}}
\drawRect{fillcolor=black!20!white,fillstyle=solid,linewidth=0.5pt}{0}{0}{0.3}{1}
\drawRect{fillcolor=black!20!white,fillstyle=solid,linewidth=0.5pt}{0.3}{0}{0.3}{1}
\drawRect{fillcolor=black!40!white,fillstyle=solid,linewidth=0.5pt}{0.6}{0}{0.4}{1}
\drawRect{}{0}{0}{1}{1}
\drawRect{fillcolor=black!40!white,fillstyle=solid,linewidth=0.5pt}{1.2}{0}{0.4}{1}
\drawRect{fillcolor=black!50!white,fillstyle=solid,linewidth=0.5pt}{1.6}{0}{0.6}{1}
\drawRect{}{1.2}{0}{1}{1}
\drawRect{fillcolor=black!50!white,fillstyle=solid,linewidth=0.5pt}{2.4}{0}{0.6}{1}
\drawRect{fillcolor=black!20!white,fillstyle=solid,linewidth=0.5pt}{3.0}{0}{0.3}{1}
\drawRect{}{2.4}{0}{1}{1}
\rput[c](0,1.3){\psline{|-|}(0.0,0)(0.3,0) \cnode*(0.15,0){2.5pt}{p11}}
\rput[c](0.3,1.3){\psline{|-|}(0.0,0)(0.3,0) \cnode*(0.15,0){2.5pt}{p12}}
\rput[c](0.6,1.3){\psline{|-|}(0.0,0)(0.4,0) \cnode*(0.2,0){2.5pt}{p13}}
\rput[c](1.2,1.3){\psline{|-|}(0.0,0)(0.4,0) \cnode*(0.2,0){2.5pt}{p21}}
\rput[c](1.6,1.3){\psline{|-|}(0.0,0)(0.6,0) \cnode*(0.3,0){2.5pt}{p22}}
\rput[c](2.4,1.3){\psline{|-|}(0.0,0)(0.6,0) \cnode*(0.3,0){2.5pt}{p31}}
\rput[c](3.0,1.3){\psline{|-|}(0.0,0)(0.3,0) \cnode*(0.15,0){2.5pt}{p32}}
\nccurve[nodesepA=1pt,nodesepB=1pt,angleA=-20,angleB=135]{->}{C1}{p31} \naput[labelsep=0pt,npos=0.3]{$\frac{1}{2}$}
\nccurve[nodesepA=1pt,nodesepB=1pt,angleA=-20,angleB=135]{->}{C1}{p22} \nbput[labelsep=0pt,npos=0.05]{$\frac{1}{2}$}
\ncline[nodesepA=1pt,nodesepB=1pt,angleA=-20,angleB=135]{->}{C2}{p13} \nbput[labelsep=0pt,npos=0.8]{$\frac{1}{2}$}
\ncline[nodesepA=1pt,nodesepB=1pt,angleA=-20,angleB=135]{->}{C2}{p21} \nbput[labelsep=0pt,npos=0.8]{$\frac{1}{2}$}
\nccurve[nodesepA=1pt,nodesepB=1pt,angleA=-135,angleB=45]{->}{C3}{p11}\nbput[labelsep=0pt,npos=0.9]{$\frac{1}{2}$}
\nccurve[nodesepA=1pt,nodesepB=1pt,angleA=-120,angleB=45]{->}{C3}{p12}\nbput[labelsep=0pt,npos=0.95]{$\frac{1}{2}$}
\nccurve[nodesepA=1pt,nodesepB=1pt,angleA=-45,angleB=90]{->}{C3}{p32}\naput[labelsep=0pt,npos=0.5]{$\frac{1}{2}$}
\rput[c](0.5,-0.5){pattern $p_1$: $x_{p_1} = \frac{1}{2}$}
\rput[c](1.7,-0.5){pattern $p_2$: $x_{p_2} = \frac{1}{2}$}
\rput[c](2.9,-0.5){pattern $p_3$: $x_{p_3} =  \frac{1}{2}$}
\drawRect{fillstyle=vlines,linewidth=0.5pt}{0.5}{8}{0.3}{1}
\drawRect{fillstyle=hlines,linewidth=0.5pt}{1.6}{8}{0.2}{1}
\drawRect{fillstyle=crosshatch,linewidth=0.5pt}{2.6}{8}{0.1}{1}
\drawRect{fillstyle=vlines,linewidth=0.5pt}{0.4}{4}{0.3}{1}
\drawRect{fillstyle=hlines,linewidth=0.5pt}{0.7}{4}{0.2}{1}
\drawRect{fillstyle=crosshatch,linewidth=0.5pt}{0.9}{4}{0.1}{1}
\drawRect{fillstyle=vlines,linewidth=0.5pt}{1.6}{4}{0.3}{1}
\drawRect{fillstyle=crosshatch,linewidth=0.5pt}{1.9}{4}{0.1}{1}
\drawRect{fillstyle=hlines,linewidth=0.5pt}{2.6}{4}{0.2}{1}
\drawRect{fillstyle=crosshatch,linewidth=0.5pt}{2.8}{4}{0.1}{1}
\rput[c](0.5,7.7){\psline{|-|}(0.0,0)(0.3,0) \cnode*(0.15,0){2.5pt}{i1}}
\rput[c](1.6,7.7){\psline{|-|}(0.0,0)(0.2,0) \cnode*(0.1,0){2.5pt}{i2}}
\rput[c](2.6,7.7){\psline{|-|}(0.0,0)(0.1,0) \cnode*(0.05,0){2.5pt}{i3}}
\rput[c](0.4,5.3){\psline{|-|}(0.0,0)(0.3,0) \cnode*(0.15,0){2.5pt}{i11}}
\rput[c](0.7,5.3){\psline{|-|}(0.0,0)(0.2,0) \cnode*(0.1,0){2.5pt}{i12}}
\rput[c](0.9,5.3){\psline{|-|}(0.0,0)(0.1,0) \cnode*(0.05,0){2.5pt}{i13}}
\rput[c](1.6,5.3){\psline{|-|}(0.0,0)(0.3,0) \cnode*(0.15,0){2.5pt}{i21}}
\rput[c](1.9,5.3){\psline{|-|}(0.0,0)(0.1,0) \cnode*(0.05,0){2.5pt}{i22}}
\rput[c](2.6,5.3){\psline{|-|}(0.0,0)(0.2,0) \cnode*(0.10,0){2.5pt}{i31}}
\rput[c](2.8,5.3){\psline{|-|}(0.0,0)(0.1,0) \cnode*(0.05,0){2.5pt}{i32}}
\ncline[nodesepA=1pt,nodesepB=1pt,angleA=-20,angleB=135]{->}{i1}{i11} \nbput[labelsep=0pt,npos=0.8]{$1$}
\nccurve[nodesepA=1pt,nodesepB=1pt,angleA=-40,angleB=135]{->}{i1}{i21} \nbput[labelsep=0pt,npos=0.8]{$1$}
\nccurve[nodesepA=1pt,nodesepB=1pt,angleA=-135,angleB=90]{->}{i2}{i12} \nbput[labelsep=0pt,npos=0.2]{$1$}
\nccurve[nodesepA=1pt,nodesepB=1pt,angleA=-135,angleB=45]{->}{i3}{i13} \nbput[labelsep=0pt,npos=0.2]{$1$}
\nccurve[nodesepA=1pt,nodesepB=1pt,angleA=-135,angleB=60]{->}{i3}{i22} \naput[labelsep=0pt,npos=0.6]{$1$}
\nccurve[nodesepA=1pt,nodesepB=1pt,angleA=-90,angleB=90]{->}{i3}{i31} \nbput[labelsep=0pt,npos=0.5]{$2$}
\nccurve[nodesepA=1pt,nodesepB=1pt,angleA=-45,angleB=70]{->}{i3}{i32} \naput[labelsep=0pt,npos=0.5]{$2$}
\drawRect{fillcolor=black!50!white,fillstyle=none,linewidth=1.5pt}{0.4}{4}{0.6}{1}
\drawRect{fillcolor=black!40!white,fillstyle=none,linewidth=1.5pt}{1.6}{4}{0.4}{1}
\drawRect{fillcolor=black!20!white,fillstyle=none,linewidth=1.5pt}{2.6}{4}{0.3}{1}
\rput[c](0.2,4.8){cont. $C_1$} \rput[c](0.2,4.2){$y_{C_1} = 1$}
\rput[c](1.4,4.8){cont. $C_2$}\rput[c](1.4,4.2){$y_{C_2} = 1$}
\rput[c](2.4,4.8){cont. $C_3$}\rput[c](2.4,4.2){$y_{C_3} = 2$}
\rput[c](0.3,8.7){item $i_1$} \rput[c](0.3,8.2){$b_1 = 2$}
\rput[c](1.4,8.7){item $i_2$} \rput[c](1.4,8.2){$b_2 = 1$}
\rput[c](2.4,8.7){item $i_3$} \rput[c](2.4,8.2){$b_3 = 7$}
\end{pspicture}
\caption{Example for assigning items to containers and containers to patterns.\label{fig:AssignmentExample}}
\end{center}
\end{figure}
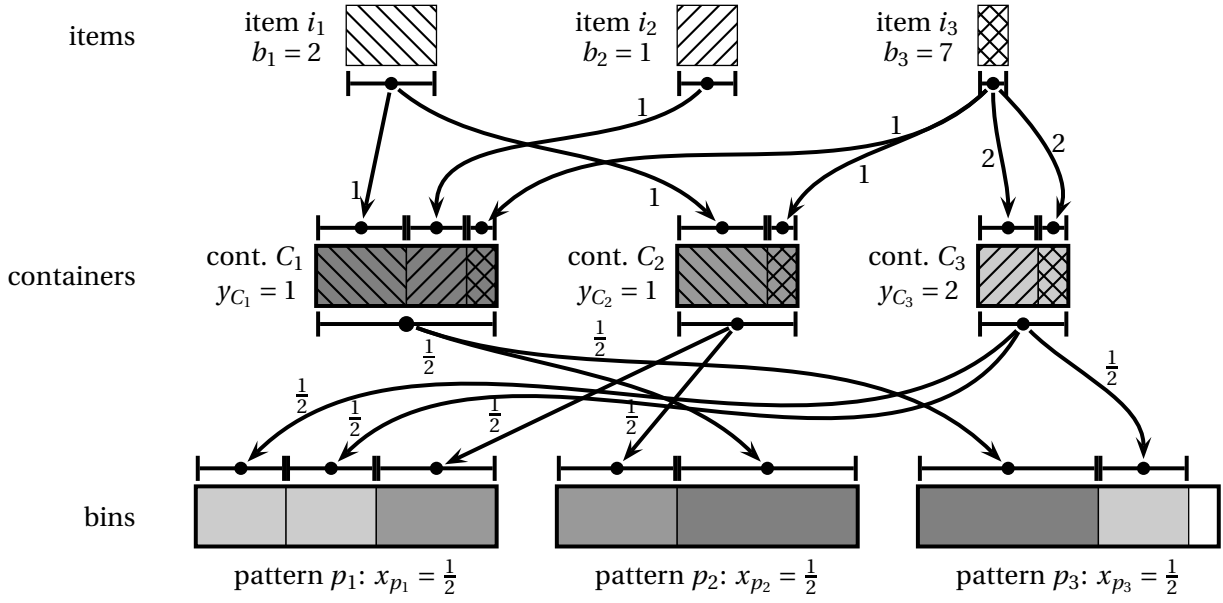
The reader might have noticed that we do allow that some copies of item $i_3$ are assigned to
slots of a larger
item $i_2$. On the other hand, we have $b_3=7$ copies of item 3, but only 6 slots in containers that
we could use. So there will be 1 unit that we won't be able to pack. Similarly, we have $y_{C_3}=2$
copies of container $C_3$, but only $\frac{3}{2}$ slots in the patterns. Later we will say that 
the \emph{deficiency} of the 2-stage packing is $1 \cdot s_3 + \frac{1}{2} \cdot s(C_3)$ where $s(C_3)$ is 
the size of container $C_3$.

Now, we want to give the formal definitions. 
We call any vector $C \in \setZ_{\geq 0}^n$ with $\sum_{i=1}^n s_iC_i \leq 1$ a \emph{container}. Here  $C_i$ denotes
the number of copies of item $i$ that are in the container. The \emph{size} of the container is
denoted by $s(C) := \sum_{i=1}^n C_is_i$. Let $\pazocal{C} := \{ C \in \setZ_{\geq 0}^n \mid s(C) \leq 1\}$ be the set of all containers. 
As we will pack containers into bins, we want to define
a \emph{pattern} as a vector of the form $p \in \setZ_{\geq 0}^{\pazocal{C}}$ where
 $p_C$ denotes the number of times that the pattern contains container $C$. 
Of course the sum of the sizes of the containers should be at most 1, thus 
\[
  \pazocal{P} := \Big\{ p \in \setZ_{\geq 0}^{\pazocal{C}} \mid \sum_{C \in \pazocal{C}} p_C \cdot s(C) \leq 1\Big\}
\]
is set of all \emph{(valid) patterns}.

Now suppose we have an instance $(s,b)$ and a fractional vector  $x \in \mathbb{R}_{\geq 0}^{\pazocal{P}}$.
To keep track of which containers should be used in the intermediate packing step, 
we also need to maintain an integral vector $y \in \mathbb{Z}_{\geq 0}^{\pazocal{C}}$. 

We say that a bipartite graph $G=(V_\ell\cup V_r,E)$ is a \emph{packing graph} if each $v\in V_\ell\cup V_r$ has an associated size $s(v)\in [0,1]$ and multiplicity $\text{mult}(v)\in \setR_{\ge 0}$, and the edge set is given by $E=\{(u,v)\in V_\ell\times V_r \mid s(u)\le s(v)\}$.
An \emph{assignment} in a packing graph is a function $a:E\rightarrow \setR_{\ge 0}$ so that for any $v\in V$, we have $\sum_{e\in \delta(v)}a(e)\le \text{mult}(v),$ where $\delta(v)$ denotes the set of edges incident to $v$.
The \emph{deficiency} of a packing graph is the total size of left nodes that fail to be packed in an optimal assignment. That is, 
$$\text{def}(G):=\min_{a \text{ assignment of } G}  \Big\{ \sum_{v\in V_\ell} s(v)\cdot(\text{mult}(v) - a(\delta(v)))  \Big\}.$$
The edge set of those graphs is extremely simple, so that one can directly obtain 
the deficiency as follows:
\begin{observation} \label{obs:GreedyAssignment}
For any packing graph, an optimal assignment $a : E \to \setR_{\geq 0}$ which attains $\textrm{def}(G)$ can be obtained
as follows: go through the nodes $v \in V_r$ in any order. Take the node $u \in V_{\ell}$
of maximum size that has some capacities left and satisfies $s(u) \leq s(v)$. Increase $a(u,v)$ as much as possible. 
\end{observation}
In this paper we further restrict ourselves to \emph{left-integral} packing graphs --- that is, for any $v\in V_\ell$, $\text{mult}(v)\in \setZ_{\ge 0}$.
We construct two packing graphs: one responsible for the assignment of items to containers and one for
assigning containers to bins.
\begin{itemize}
\item \emph{Assigning items to containers:} Given $b\in\setZ_{\ge 0}^n$,$y\in\setZ_{\ge 0}^\pazocal{C}$, we define a packing graph $G_1(b,y)$ as follows. 
The left nodes of the graph are defined by $V_\ell=[n]$, with sizes $s_i$ and multiplicities $b_i$.
The right nodes are defined by $V_r = \{ (i,C) : C \in \pazocal{C}, i \in C \}$ with the size of node $(i,C)$ given by $s_i$ and multiplicity by $y_C\cdot C_i$.

\item \emph{Assigning containers to patterns:} Given $y\in\setZ_{\ge 0}^\pazocal{C}$ and $x\in\setR_{\ge 0}^\pazocal{P}$, we define a packing graph $G_2(x,y)$. The left nodes are given by  $\pazocal{C}$ with sizes given by the sizes of containers, and multiplicities $y_C$.
 The right nodes are given by $ \{ (C,p) : C \in \pazocal{C}, p \in \pazocal{P} \}$, with the size of node $(C,p)$ given by $s(C)$ and the multiplicity by
$x_C \cdot p_{C}$. 
\end{itemize}
We then define the deficiency of the pair $(x,y)$ with item multiplicities $b$ to be the sum
$$\textrm{def}_b(x,y) := \textrm{def}(G_1(b,y)) + \textrm{def}(G_2(x,y)).$$ 
In later sections we will often leave off the $b$ to simplify notation.

We should discuss why the 2-stage packing via the containers is useful. 
First of all, it is easy to find \emph{some} initial configuration. 
\begin{lemma} \label{lem:StartingSolution}
For any bin packing instance $(s,b)$, one can compute a ``starting solution''  $x \in \mathbb{R}_{\geq 0}^{\pazocal{P}}$
and $y \in \mathbb{Z}_{\geq 0}^{\pazocal{C}}$
 in polynomial time so that $\bm{1}^Tx \leq OPT_f + 1$ and $\textrm{def}(x,y) = 0$ with $|\textrm{supp}(x)| \leq n$.
\end{lemma}
\begin{proof}
As we already argued, one can compute a 
 fractional solution $x$ for \eqref{eq:GilmoreGomory} in polynomial time that 
has cost $\bm{1}^Tx \leq OPT_f + 1$. We simply use singleton containers $\{i\}$ for all items $i \in \{ 1,\ldots,n\}$ 
and set $y_{\{i\}} := b_i$. 
\end{proof}

Next, we argue that our notation of deficiency was actually meaningful in recovering an assignment
of items to bins. 
\begin{lemma} \label{lem:PackingItemsWith2DefItems}
Suppose that $x \in \setZ_{\geq 0}^{\pazocal{P}},y\in \setZ_{\geq 0}^{\pazocal{C}}$ are both integral. 
Then there is a 
packing of all items into at most $\bm{1}^Tx + 2\textrm{def}(x,y) + 1$ bins. 
\end{lemma}
\begin{proof}
Since $x$ and $y$ are both integral, all multiplicities
in $G_1$ and $G_2$ will be integral and we can find two integral assignments $a_1,a_2$ attaining $\textrm{def}(x,y)$.
 Buy all the patterns suggested by $x$. Use $a_2$ to pack the containers in $y$. Then use $a_1$ to map the items to 
containers. There are some items that will not be assigned --- their total size is $\textrm{def}(G_1(b,y))$. 
Moreover, there might also
be containers in $y$ that have not been assigned; their total size is $\textrm{def}(G_2(x,y))$. 
We pack items and containers greedily into at most $2\textrm{def}(x,y) + 1$ many extra bins using Lemma~\ref{lem:GreedyPacking}. 
\end{proof}

In each iteration of our algorithm, it will be useful for us to be able fix the integral part of $x$ and focus solely on the fractional part.

\begin{lemma} \label{lem:SplittingOffFractionalPart}
Suppose $x\in \setR_{\geq 0}^\pazocal{P}, y\in\setZ_{\ge 0}^\pazocal{C}$, and 
$b\in \setZ_{\ge 0}^n$. 
If $\hat{x}_p=\lfloor x_p \rfloor$ for all patterns $p$, 
then there exist vectors $\hat{y}\in\setZ_{\ge 0}^\pazocal{C}$, 
$\hat{b}\in\setZ_{\ge 0}^n$ with $\hat{b} \leq b$ so that 
$\text{def}_{\hat{b}}(\hat{x},\hat{y})=0$ and 
$\text{def}_{b-\hat{b}}(x-\hat{x},y-\hat{y})=\text{def}_b(x,y)$. 
\end{lemma}
\begin{proof}
Let us imagine that we replace each node $(C,p)$ in $G_2(x,y)$ with two copies, a ``red'' node
and a ``blue'' node. 
The red copy receives an integral multiplicity of $\text{mult}_\text{red}(C,p)=\hat{x}_p p_C$
while the blue copy receives a fractional multiplicity of 
$\text{mult}_\text{blue}(C,p)=(x_p-\hat{x}_p)\cdot p_C$. Now we apply 
Observation~\ref{obs:GreedyAssignment} to find the best assignment $a$. Crucially, 
we set up the order of the right hand side nodes so that we first process the red integral
nodes and then the blue fractional ones. Note that the assignment that this greedy procedure
computes is optimal and moreover, the assignments for red nodes will be integral. 
For each container $C$ on the left, we define 
$\hat{y}_C$ to be the total red multiplicity of its targets under this optimal assignment.
Then $\text{def}(G_2(\hat{x},\hat{y}))=0$ and $\text{def}(G_2(x-\hat{x},y-\hat{y}))=\text{def}(G_2(x,y))$.
In the graph $G_1(b,y)$, all multiplicities are integral anyway, so 
we can trivially find an integral vector $\hat{b}$ so that 
$\text{def}(G_1(\hat{b},\hat{y}))=0$ and $\text{def}(G_1(b-\hat{b},y-\hat{y}))=\text{def}(G_1(b,y))$.
\end{proof}

Define $\text{supp}(x) := \{ p \in \pazocal{P} : x_p > 0\}$ as the support of $x$ and
$\textrm{frac}(x) := \{ p \in \pazocal{P}: 0<x_p<1\}$ as the patterns in $p$ that
are still fractional.
Now we have enough notation to state our main technical theorem: 
\begin{theorem} \label{thm:OneIteration}
Let $(s,b)$ be an instance with $s_1,\ldots,s_n \geq s_{\min}>0$. Let $y \in \setZ_{\geq 0}^{\pazocal{C}}$
and $x \in [0,1[^{\pazocal{P}}$ with $|\text{supp}(x)|\geq L\log(\frac{1}{s_{\min}})$, where $L$ is a large enough constant. 
Then there is a randomized polynomial time algorithm 
that finds $\tilde{y} \in \setZ_{\geq 0}^{\pazocal{C}}$ and $\tilde{x} \in \setR_{\geq 0}^{\pazocal{P}}$
with $\bm{1}^T\tilde{x} = \bm{1}^Tx$ and $\textrm{def}(\tilde{x},\tilde{y}) \leq \textrm{def}(x,y) + O(1)$ while $|\textrm{frac}(\tilde{x})| \leq \frac{1}{2}|\textrm{frac}(x)|$.
\end{theorem}
While it will take the remainder of this paper to prove the theorem, the algorithm behind the statement
can be split into the following two steps: 
\begin{enumerate}
\item[(I)] \emph{Rebuilding the container assignment}: We will change the assignments for the 
pair $(x,y)$ so that for every container in size class $\sigma$ the patterns in supp$(x)$ use, they use nearly $(\frac{1}{\sigma})^{1/2}$ copies, while no individual pattern in $\textrm{supp}(x)$ contains more than
$(\frac{1}{\sigma})^{1/4}$ copies of the same container. 
\item[(II)] \emph{Application of Lovett-Meka}: We will apply the Lovett-Meka algorithm to sparsify the fractional solution $x$. Here, the vectors $v_j$ that comprise the input for the LM-algorithm will correspond to
sums over intervals of rows of the constraint matrix $A$. Recall that the error bound provided by Lovett-Meka crucially depends on the lengths $\|v_j\|_2$. The procedure in $(I)$ will ensure 
that the Euclidean length of those vectors is small. 
\end{enumerate}

Once we have proven Theorem~\ref{thm:OneIteration}, the main result easily follows: 
\begin{proof}[Proof of Theorem~\ref{thm:MainContribution}]
We compute a fractional solution $x$ to \eqref{eq:GilmoreGomory} of cost $\bm{1}^Tx \leq OPT_f+1$. 
In fact, we can assume that $x$ is a basic solution to the LP and hence $|\textrm{supp}(x)| \leq n$. 
We construct a container assignment $y$ consisting only of singletons, 
see Lemma~\ref{lem:StartingSolution}.
Then for $\log(n)$ iterations, we first use Lemma~\ref{lem:SplittingOffFractionalPart} 
to split the current solution $x$ as $x = x^{\textrm{int}} + x^{\textrm{frac}}$ where $x^{\textrm{int}}_p = \lfloor x_p \rfloor$ and obtain a corresponding split $y = y^{\textrm{int}} + y^{\textrm{frac}}$.
Then we run Theorem~\ref{thm:OneIteration} with input $(x^{\textrm{frac}},y^{\textrm{frac}})$
and denote the result by $(\tilde{x}^{\textrm{frac}},\tilde{y}^{\textrm{frac}})$. Finally we 
update $x := x^{\textrm{int}} + \tilde{x}^{\textrm{frac}}$ and $y := y^{\textrm{int}} + \tilde{y}^{\textrm{frac}}$.

As soon as $|\textrm{frac}(x)| \leq O(\log \frac{1}{s_{\min}})$, we can just buy every pattern in frac$(x)$. 
In each iteration the deficiency increases by at most $O(1)$. At the end, we use Lemma~\ref{lem:PackingItemsWith2DefItems} 
to actually pack the items into bins. We arrive at a solution of cost $OPT_f + O(\log \max\{ n, \frac{1}{s_{\min}} \})$
which is enough, using Lemma~\ref{lem:AssumptionSizesAtLeast1-n}.
\end{proof}
We will describe the implementation of $(I)$ in Section~\ref{sec:RebuildingContainers} and
then $(II)$ in Section~\ref{sec:ApplyingLM}.

\section{Rebuilding the container assignment\label{sec:RebuildingContainers}}
In this section we assume that we are given $x\in [0,1[^\pazocal{P}$ with $|\text{supp}(x)|=m$. 
To ease notation, we will only write the nonzero parts of $x$, so that if supp$(x)=\{p_1,p_2,...,p_m\}$, then $x=(x_{p_i})_{i=1}^m$. We update $x$ by altering the patterns that make up its support. Even though some patterns could become identical, we continue to treat them as separate patterns.

Originally, we had defined $A$ as the incidence matrix of the Gilmore Gomory LP in \eqref{eq:GilmoreGomory}
where the rows correspond to items. Due to our 2-stage packing, we actually consider the patterns to be
multi-sets of containers, not items anymore. Hence, let us for the rest of the paper redefine the
meaning of $A$. 
Now, the rows of $A$ correspond to the containers in $\pazocal{C}$ ordered from largest to smallest, and columns represent the patterns in supp$(x)$. As we perform the grouping and container-forming operations, we update the columns of the matrix. The resulting columns then yield a new fractional solution $\tilde{x}$ by taking $x_{p_i}$ copies of the pattern now in column $i$.

We will now describe our grouping and container reassignment operations, keeping track of what happens to the fractional solution as well as to the corresponding matrix.

First, we need a lemma that tells us how rebuilding the fractional solution affects the deficiency. 
To have some useful notation, define $\text{mult}(C,x):=\sum_{p\in\pazocal{P}}\text{mult}(C,p)=\sum_{p\in\pazocal{P}}x_pp_C$ to be the number of times that the patterns cover container $C \in \pazocal{C}$. 

Now, if $\sum_{s(C)\ge s}\tilde{y}_C\le \sum_{s(C)\ge s}y_C$ for all $s\ge 0$, then we write $\tilde{y} \preceq y$.
Moreover, if $\sum_{s(C)\ge s}\text{mult}(C,\tilde{x})\ge \sum_{s(C)\ge s}\text{mult}(C,x)$ for all $s \geq 0$, then we write $\tilde{x}\succeq x$.
Observe that if $\tilde{y}\preceq y$ and $\tilde{x}\succeq x$, then $\text{def}(G_2(\tilde{x},\tilde{y})) \leq \text{def}(G_2(x,y))$. 
\begin{lemma} \label{lem:DefIncrease}Now suppose that $t_\sigma\ge 0$ is such that 
$$\sum_{s(C)\ge s}\text{mult}(C,\tilde{x})
\ge\left\{\begin{array}{lll}\sum_{s(C)\ge s}\text{mult}(C,x) & \text{ if } & s> \sigma\\
 \sum_{s(C)\ge s}\text{mult}(C,x)-t_\sigma & \text{ if } & s\le \sigma\end{array}\right.$$
Then $\text{def}(\tilde{x},y)\le \text{def}(x,y)+\sigma \cdot t_\sigma$.
\end{lemma}
\begin{proof}
Let $C_0$ be the largest container of size at most $\sigma$, and let $x'$ be the vector representing $t_\sigma$ copies of the pattern containing a single copy of $C_0$.
Then $\tilde{x}+x'\succeq x$, and so $\text{def}(G_2(\tilde{x}+x',y))\le\text{def}(G_2(x,y))$.
But if $y'$ is the vector representing $t_\sigma$ copies of $C_0$, then $\text{def}(G_2(\tilde{x},y))=\text{def}(G_2(\tilde{x}+x',y+y'))$, since we can find an optimal assignment taking the containers in $y'$ to those of $x'$.
Since the total size of $y'$ is at most $\sigma t_\sigma$, we have $\text{def}(G_2(\tilde{x},y))\le \text{def}(G_2(\tilde{x}+x',y))+\sigma t_\sigma \le \text{def}(G_2(x,y))+ \sigma t_\sigma$, and therefore $\text{def}(\tilde{x},y)\le \text{def}(x,y)+\sigma t_\sigma$.
\end{proof}

If $\sigma$ is a power of 2, say $\sigma = 2^{-\ell}$ for $\ell \in \setZ_{\geq 0}$, then we say the \emph{size class}
of $\sigma$ is the set of items with sizes between $\frac{1}{2}\sigma$ and $\sigma$.  
In this next lemma, we round containers in patterns down so that each container type in size class $\sigma$ is either not used at all or is used at least $\frac{\delta}{\sigma}$ times.

\begin{lemma}[Grouping]
Let  $(s,b)$ be a bin packing instance with $y \in \mathbb{Z}_{\geq 0}^{\pazocal{C}}$ and $x \in [0,1[^{\pazocal{P}}$. For any size class $\sigma$ and $\delta>0$, we can find $\tilde{x} \in [0,1[^{\pazocal{P}}$
so that
\begin{enumerate}
\item $\bm{1}^T\tilde{x} = \bm{1}^Tx$
\item $|\textrm{supp}(\tilde{x})| \le |\textrm{supp}(x)|$
\item For each container type $C$ in size class $\sigma$, either $\text{mult}(C,\tilde{x})=0$ or $s(C)\cdot\text{mult}(C,\tilde{x})\geq \delta$. In all other size classes, the multiplicities of containers in patterns do not change. 

\item $\textrm{def}(\tilde{x},y) \leq \textrm{def}(x,y) + O(\delta)$.
\end{enumerate}
\end{lemma}

\begin{proof}
Assume containers are sorted by size, from largest to smallest.
Define $S_\delta$ to be the set of containers in size class $\sigma$ not satisfying condition (3) above. In other words, \\$S_\delta:=\{C \text{ in size class } \sigma\mid  0<s(C)\cdot\text{mult}(C,x)<\delta\}$.

For a subset $H\subset S_\delta$, define the weight of $H$ to be $w(H):=\sum_{C\in H} s(C)\cdot \text{mult}(C,x)$.
Note that the weight of a single container is at most $\delta.$
Hence we can partition $S_\delta=H_1\dot\cup H_2 \dot\cup ... \dot\cup H_r$ so that:

\begin{enumerate}
\item $w(H_k)\in[2\delta,3\delta], \forall k=1,...,r-1$.
\item $w(H_r)\le 3\delta$.
\item $C\in H_k, C'\in H_{k+1} \textrm{ implies } s(C)\ge s(C')$.
\end{enumerate}

For each $k=1,...,r-1$ and container $C \in H_k$, we replace containers of type $C$ in all patterns $p \in \text{frac}(x)$ 
with the smallest container type appearing in $H_k$. 
For all $C \in H_r$, remove containers of type $C$ from all patterns $p \in \text{frac}(x)$.
Call the updated vector $\tilde{x}$. We see immediately that $\bm{1}^T\tilde{x} = \bm{1}^Tx$ and $|\textrm{supp}(\tilde{x})|\le |\textrm{supp}(x)|$. 

Moreover, since every container type $C$ appearing in $\tilde{x}$ now has an entire group using it, and the weight of each container didn't change by more than a factor of $2$, we have $s(C)\cdot \text{mult}(C,\tilde{x})\ge \delta$, and so condition (3) is satisfied.
To complete the proof, it remains to show that $\textrm{def}(G_2(\tilde{x},y))\le \textrm{def}(G_2(x,y)) + O(\delta)$.

Now, for any $i$, there is at most one group $H_k$ whose containers (partly) changed from being larger than $s(C_i)$ to smaller. The weight of this group is at most $3\delta$, and so $\sum_{j\le i}\text{mult}(C_j,x)-\sum_{j\le i} \text{mult}(C_j,\tilde{x})\le\frac{6\delta}{\sigma}$.
Since this holds for all $i$, we can therefore apply Lemma~\ref{lem:DefIncrease} to conclude that $\textrm{def}(G_2(\tilde{x},y))\le \textrm{def}(G_2(x,y)) + O(\delta)$.
\end{proof}

We now remark what happens to the associated matrix $A$ under this grouping operation. Write $A,\tilde{A}$ as our original and updated matrices, and $A_C,\tilde{A}_C$ as the rows for container $C$. 
For container types $C$ in size class $\sigma$, either $\tilde{A}_{C}x=0$ or $s(C)\cdot \tilde{A}_{C}x\ge \delta$.
For all other size classes, $\tilde{A}_{C}=A_{C}$.
In particular, notice that we have either $\|\tilde{A}_C\|_1=0$ or $s(C) \cdot \|\tilde{A}_C\|_1 \ge \delta.$

Before we introduce the next main lemma --- how to reassign containers --- we prove a useful result about decomposing packing graphs in a nice way. For a visualization of the following lemma, 
see Figure~\ref{fig:ExampleDecomposition}. 
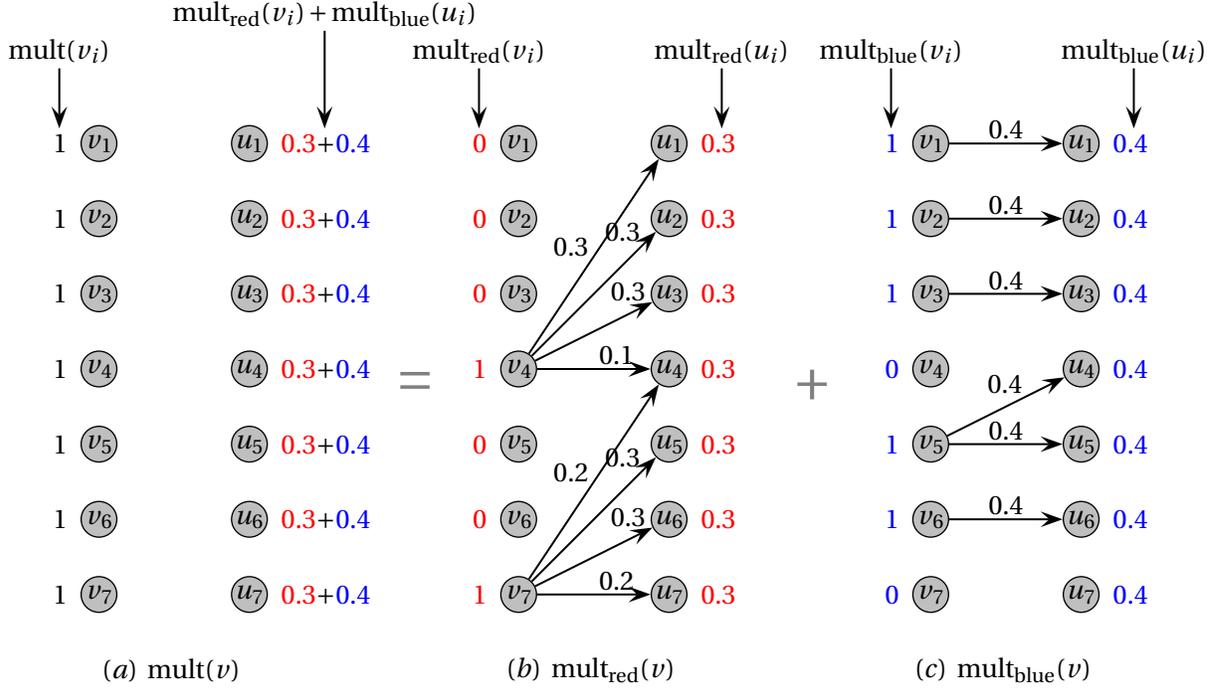
\begin{figure}
\begin{center}
\begin{pspicture}(0,0)(5.5,8)
\multido{\N=7+-1,\n=1+1}{7}{%
\cnode[linewidth=0.5pt,fillstyle=solid,fillcolor=lightgray](0,\N){7pt}{v\n}
\rput[c](v\n){$v_{\n}$}
\cnode[linewidth=0.5pt,fillstyle=solid,fillcolor=lightgray](2,\N){7pt}{u\n}
\rput[c](u\n){$u_{\n}$}
\nput[labelsep=5pt]{0}{u\n}{$\red{0.3}\black{+}\blue{0.4}$}
\nput[labelsep=5pt]{180}{v\n}{$1$}
}
\rput[c](4.2,3.8){\Huge{\gray{=}}}
\rput[c](1,0){$(a) \; \; \textrm{mult}(v)$ }
\pnode(-15pt,7.2){A} \pnode(-15pt,8){B} \ncline[linewidth=0.75pt]{->}{B}{A} \nput{90}{B}{$\textrm{mult}(v_i)$}
\pnode(3,7.2){C} \pnode(3,8.5){D} \ncline[linewidth=0.75pt]{->}{D}{C} \nput{90}{D}{$\textrm{mult}_{\textrm{red}}(v_i)+\textrm{mult}_{\text{blue}}(u_i)$}
\end{pspicture}
\begin{pspicture}(0,0)(5.4,8)
\multido{\N=7+-1,\n=1+1}{7}{%
\cnode[linewidth=0.5pt,fillstyle=solid,fillcolor=lightgray](0,\N){7pt}{v\n}
\rput[c](v\n){$v_{\n}$}
\cnode[linewidth=0.5pt,fillstyle=solid,fillcolor=lightgray](2,\N){7pt}{u\n}
\rput[c](u\n){$u_{\n}$}
\nput[labelsep=5pt]{0}{u\n}{$\red{0.3}$}
}
\nput[labelsep=5pt]{180}{v1}{\red $0$}
\nput[labelsep=5pt]{180}{v2}{\red $0$}
\nput[labelsep=5pt]{180}{v3}{\red $0$}
\nput[labelsep=5pt]{180}{v4}{\red $1$}
\nput[labelsep=5pt]{180}{v5}{\red $0$}
\nput[labelsep=5pt]{180}{v6}{\red $0$}
\nput[labelsep=5pt]{180}{v7}{\red $1$}
\ncline[linewidth=0.75pt]{->}{v4}{u1} \naput[labelsep=1pt,npos=0.5]{$0.3$}
\ncline[linewidth=0.75pt]{->}{v4}{u2} \naput[labelsep=1pt,npos=0.9]{$0.3$}
\ncline[linewidth=0.75pt]{->}{v4}{u3} \naput[labelsep=1pt,npos=0.9]{$0.3$}
\ncline[linewidth=0.75pt]{->}{v4}{u4} \naput[labelsep=1pt,npos=0.7]{$0.1$}
\ncline[linewidth=0.75pt]{->}{v7}{u4} \naput[labelsep=1pt,npos=0.5]{$0.2$}
\ncline[linewidth=0.75pt]{->}{v7}{u5} \naput[labelsep=1pt,npos=0.9]{$0.3$}
\ncline[linewidth=0.75pt]{->}{v7}{u6} \naput[labelsep=1pt,npos=0.9]{$0.3$}
\ncline[linewidth=0.75pt]{->}{v7}{u7} \naput[labelsep=1pt,npos=0.7]{$0.2$}
\rput[c](3.9,3.8){\Huge{\gray{+}}}
\rput[c](1,0){$(b) \; \; \textrm{mult}_{\textrm{red}}(v)$ }
\pnode(-15pt,7.2){A} \pnode(-15pt,8){B} \ncline[linewidth=0.75pt]{->}{B}{A} \nput{90}{B}{$\textrm{mult}_{\textrm{red}}(v_i)$}
\pnode(2.7,7.2){C} \pnode(2.7,8){D} \ncline[linewidth=0.75pt]{->}{D}{C} \nput{90}{D}{$\textrm{mult}_{\textrm{red}}(u_i)$}
\end{pspicture}
\begin{pspicture}(0,0)(2,8)
\multido{\N=7+-1,\n=1+1}{7}{%
\cnode[linewidth=0.5pt,fillstyle=solid,fillcolor=lightgray](0,\N){7pt}{v\n}
\rput[c](v\n){$v_{\n}$}
\cnode[linewidth=0.5pt,fillstyle=solid,fillcolor=lightgray](2,\N){7pt}{u\n}
\rput[c](u\n){$u_{\n}$}
\nput[labelsep=5pt]{0}{u\n}{$\blue{0.4}$}
}
\nput[labelsep=5pt]{180}{v1}{\blue $1$}
\nput[labelsep=5pt]{180}{v2}{\blue $1$}
\nput[labelsep=5pt]{180}{v3}{\blue $1$}
\nput[labelsep=5pt]{180}{v4}{\blue $0$}
\nput[labelsep=5pt]{180}{v5}{\blue $1$}
\nput[labelsep=5pt]{180}{v6}{\blue $1$}
\nput[labelsep=5pt]{180}{v7}{\blue $0$}
\ncline[linewidth=0.75pt]{->}{v1}{u1} \naput[labelsep=1pt,npos=0.5]{$0.4$}
\ncline[linewidth=0.75pt]{->}{v2}{u2} \naput[labelsep=1pt,npos=0.5]{$0.4$}
\ncline[linewidth=0.75pt]{->}{v3}{u3} \naput[labelsep=1pt,npos=0.5]{$0.4$}
\ncline[linewidth=0.75pt]{->}{v5}{u4} \naput[labelsep=1pt,npos=0.6]{$0.4$}
\ncline[linewidth=0.75pt]{->}{v5}{u5} \naput[labelsep=1pt,npos=0.5]{$0.4$}
\ncline[linewidth=0.75pt]{->}{v6}{u6} \naput[labelsep=1pt,npos=0.5]{$0.4$}
\rput[c](1,0){$(c) \; \; \textrm{mult}_{\textrm{blue}}(v)$ }
\pnode(-15pt,7.2){A} \pnode(-15pt,8){B} \ncline[linewidth=0.75pt]{->}{B}{A} \nput{90}{B}{$\textrm{mult}_{\textrm{blue}}(v_i)$}
\pnode(2.7,7.2){C} \pnode(2.7,8){D} \ncline[linewidth=0.75pt]{->}{D}{C} \nput{90}{D}{$\textrm{mult}_{\textrm{blue}}(u_i)$}
\end{pspicture}
\caption{Visualization of the decomposition of the left hand side multiplicities from Lemma~\ref{decomposition}. Here, $V_{\ell}=\{v_1,\ldots,v_7\}$ and $V_{r} = \{u_1,\ldots,u_7\}$. 
Assume that $s(u_i) = s(v_i)$ and $s(v_1) > \ldots > s(v_7)$. Nodes are labelled with their multiplicities. In $(b)$ and $(c)$ we also depict the assignments corresponding to the deficiencies. \label{fig:ExampleDecomposition}}
\end{center}
\end{figure}

\begin{lemma} 
 \label{decomposition} Suppose $G=(V_\ell\cup V_r,E)$ is a left-integral packing graph as in Section \ref{sec:Deficiency}, and that for every $v\in V_r$, we are given red and blue multiplicities so that $\text{mult}(v)=\text{mult}_\text{red}(v)+\text{mult}_\text{blue}(v)$.
Suppose further that all nodes $v\in V_r$ of size greater than $\sigma$ have $\text{mult}_\text{red}(v)=0$.
Then we can find left-integral packing graphs $G_\text{red}$ and $G_\text{blue}$ with the same edges, nodes, and sizes of $G$ but with multiplicities satisfying $\text{mult}_\text{red} + \text{mult}_\text{blue}=\text{mult}$. Moreover, we have 
$\text{def}(G_\text{red})=0$ and $\text{def}(G_\text{blue})\le \text{def}(G)+\sigma$.\end{lemma}

\begin{proof} 
By allowing fractional red and blue multiplicities, we can find initial values for the red and blue multiplicities of left nodes so that $\text{def}(G_\text{red})=0$ and $\text{def}(G_\text{blue})=\text{def}(G)$.
To enforce integrality, we will update these multiplicities by swapping (fractional parts of) larger red nodes for smaller blue nodes.

Suppose nodes on the left with positive red multiplicity are ordered by size, so that $\sigma \ge s(v_1)\ge s(v_2)\ge ... \ge s(v_\ell)$.
While the multiplicities are not all integral, let $i$ be the index of the largest $v_i$ with $\text{mult}_\text{red}(v_i)$ not integral. 
If $i<\ell$, decrease $\text{mult}_{\text{red}}(v_i)$ to $\lfloor\text{mult}_{\text{red}}(v_i)\rfloor$, and increase $\text{mult}_\text{red}(v_{i+1})$ by the same amount. If $i=\ell$, simply decrease $\text{mult}_\text{red}(v_\ell)$ to $\lfloor\text{mult}_\text{red}(v_\ell)\rfloor.$
Notice that the deficiency of the red graph has not increased, since we are either replacing nodes with smaller nodes or decreasing the multiplicity of the last node.
Moreover, we notice that for any size $s$, the total red multiplicity of nodes at least size $s$ has decreased by at most $1$.
Therefore in the complementary blue graph, 
$$ \sum_{v \in V_{\ell}:s(v)\ge s}(\text{mult}'_\text{blue}(v)-\text{mult}_\text{blue}(v))\le 1.$$
The additional blue nodes we fail to pack will therefore all have size at most $\sigma$ and their total multiplicity will be at most $1$, so 
the deficiency of the blue graph increases by at most $\sigma$.
\end{proof}

A key technical ingredient for our algorithm is to be able to replace sets of identical 
copies of a container in patterns of $x$ by a bigger container that contains the union of
the smaller containers. 
\begin{lemma}\label{lem:Gluing}
Given a pair $(x,y)$ with $x \in \setR_{\geq 0}^{\pazocal{P}}$ and $y \in \setZ_{\geq 0}^{\pazocal{C}}$.
Let $k \in \setN$ and $0<\sigma \leq 1$ be two parameters.  
Let $\tilde{x} \in \setR_{\geq 0}^{\pazocal{P}}$ be the vector that emerges if for all containers $C$ with $\frac{1}{2}\sigma\le 
s(C)\le \sigma$ and all patterns $p$ 
we replace $k\cdot\lfloor \frac{p_C}{k}\rfloor$ copies of $C$ by $\lfloor\frac{p_C}{k}\rfloor$ copies of the container that is $k \cdot C$. Then there is a
$\tilde{y} \in \setZ_{\geq 0}^{\pazocal{C}}$ so that $\textrm{def}(\tilde{x},\tilde{y}) \leq \textrm{def}(x,y) + O(k\sigma)$.
\end{lemma}

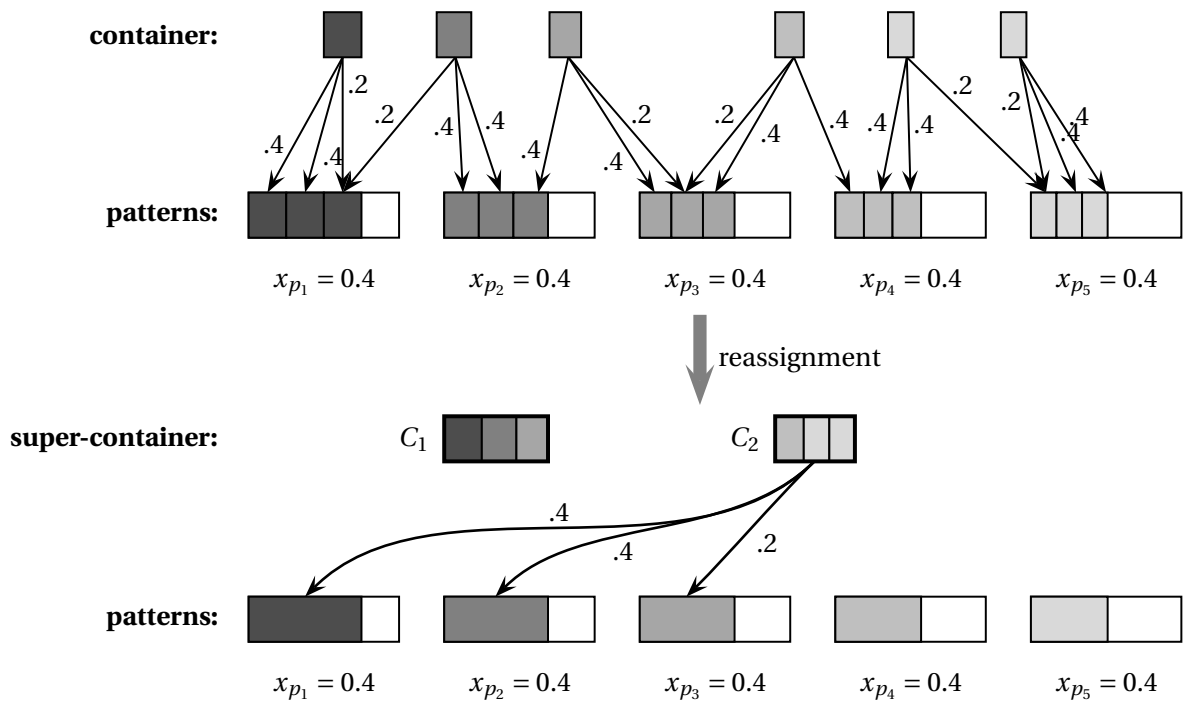
\begin{figure}
\begin{center}
\psset{xunit=2.0cm,yunit=0.6cm}
\begin{pspicture}(-1,-0.4)(8,5)
\rput[r](-0.2,4.5){{\bf container:}}
\rput[r](-0.2,0.5){{\bf patterns:}}
\rput[c](0.5,4){ 
\drawRect{linewidth=0.75pt,fillstyle=solid,fillcolor=black!70!white}{0}{0}{0.25}{1}
\pnode(0.125,0){c1}
}
\rput[c](1.25,4){ 
\drawRect{linewidth=0.75pt,fillstyle=solid,fillcolor=black!50!white}{0}{0}{0.23}{1}
\pnode(0.125,0){c2}
}
\rput[c](2.0,4){ 
\drawRect{linewidth=0.75pt,fillstyle=solid,fillcolor=black!35!white}{0}{0}{0.21}{1}
\pnode(0.125,0){c3}
}
\rput[c](3.5,4){ 
\drawRect{linewidth=0.75pt,fillstyle=solid,fillcolor=black!25!white}{0}{0}{0.19}{1}
\pnode(0.125,0){c4}
}
\rput[c](4.25,4){ 
\drawRect{linewidth=0.75pt,fillstyle=solid,fillcolor=black!15!white}{0}{0}{0.17}{1}
\pnode(0.125,0){c5}
}
\rput[c](5.0,4){ 
\drawRect{linewidth=0.75pt,fillstyle=solid,fillcolor=black!15!white}{0}{0}{0.17}{1}
\pnode(0.125,0){c6}
}
\rput[c](0,0){
\drawRect{linewidth=0.75pt}{0}{0}{1}{1}
\drawRect{linewidth=0.75pt,fillstyle=solid,fillcolor=black!70!white}{0.0}{0}{0.25}{1}
\drawRect{linewidth=0.75pt,fillstyle=solid,fillcolor=black!70!white}{0.25}{0}{0.25}{1}
\drawRect{linewidth=0.75pt,fillstyle=solid,fillcolor=black!70!white}{0.5}{0}{0.25}{1}
\pnode(0.125,1){p11}
\pnode(0.375,1){p12}
\pnode(0.625,1){p13}
\rput[c](0.5,-1){$x_{p_1} = 0.4$}
}
\rput[c](1.3,0){
\drawRect{linewidth=0.75pt}{0}{0}{1}{1}
\drawRect{linewidth=0.75pt,fillstyle=solid,fillcolor=black!50!white}{0.0}{0}{0.23}{1}
\drawRect{linewidth=0.75pt,fillstyle=solid,fillcolor=black!50!white}{0.23}{0}{0.23}{1}
\drawRect{linewidth=0.75pt,fillstyle=solid,fillcolor=black!50!white}{0.46}{0}{0.23}{1}
\pnode(0.125,1){p21}
\pnode(0.375,1){p22}
\pnode(0.625,1){p23}
\rput[c](0.5,-1){$x_{p_2} = 0.4$}
}
\rput[c](2.6,0){
\drawRect{linewidth=0.75pt}{0}{0}{1}{1}
\drawRect{linewidth=0.75pt,fillstyle=solid,fillcolor=black!35!white}{0.0}{0}{0.21}{1}
\drawRect{linewidth=0.75pt,fillstyle=solid,fillcolor=black!35!white}{0.21}{0}{0.21}{1}
\drawRect{linewidth=0.75pt,fillstyle=solid,fillcolor=black!35!white}{0.42}{0}{0.21}{1}
\pnode(0.1,1){p31}
\pnode(0.3,1){p32}
\pnode(0.5,1){p33}
\rput[c](0.5,-1){$x_{p_3} = 0.4$}
}
\rput[c](3.9,0){
\drawRect{linewidth=0.75pt}{0}{0}{1}{1}
\drawRect{linewidth=0.75pt,fillstyle=solid,fillcolor=black!25!white}{0.0}{0}{0.19}{1}
\drawRect{linewidth=0.75pt,fillstyle=solid,fillcolor=black!25!white}{0.19}{0}{0.19}{1}
\drawRect{linewidth=0.75pt,fillstyle=solid,fillcolor=black!25!white}{0.38}{0}{0.19}{1}
\pnode(0.1,1){p41}
\pnode(0.3,1){p42}
\pnode(0.5,1){p43}
\rput[c](0.5,-1){$x_{p_4} = 0.4$}
}
\rput[c](5.2,0){
\drawRect{linewidth=0.75pt}{0}{0}{1}{1}
\drawRect{linewidth=0.75pt,fillstyle=solid,fillcolor=black!15!white}{0.0}{0}{0.17}{1}
\drawRect{linewidth=0.75pt,fillstyle=solid,fillcolor=black!15!white}{0.17}{0}{0.17}{1}
\drawRect{linewidth=0.75pt,fillstyle=solid,fillcolor=black!15!white}{0.34}{0}{0.17}{1}
\pnode(0.1,1){p51}
\pnode(0.3,1){p52}
\pnode(0.5,1){p53}
\rput[c](0.5,-1){$x_{p_5} = 0.4$}
}
\ncline[linewidth=0.75pt]{->}{c1}{p11} \nbput[labelsep=1pt,npos=0.7]{$.4$}
\ncline[linewidth=0.75pt]{->}{c1}{p12} \naput[labelsep=1pt,npos=0.7]{$.4$}
\ncline[linewidth=0.75pt]{->}{c1}{p13} \naput[labelsep=1pt,npos=0.2]{$.2$}
\ncline[linewidth=0.75pt]{->}{c2}{p13} \nbput[labelsep=1pt,npos=0.5]{$.2$}
\ncline[linewidth=0.75pt]{->}{c2}{p21} \nbput[labelsep=1pt,npos=0.5]{$.4$}
\ncline[linewidth=0.75pt]{->}{c2}{p22} \naput[labelsep=1pt,npos=0.5]{$.4$}
\ncline[linewidth=0.75pt]{->}{c3}{p23} \nbput[labelsep=1pt,npos=0.7]{$.4$}
\ncline[linewidth=0.75pt]{->}{c3}{p31} \nbput[labelsep=1pt,npos=0.7]{$.4$}
\ncline[linewidth=0.75pt]{->}{c3}{p32} \naput[labelsep=1pt,npos=0.5]{$.2$}
\ncline[linewidth=0.75pt]{->}{c4}{p32} \nbput[labelsep=1pt,npos=0.5]{$.2$}
\ncline[linewidth=0.75pt]{->}{c4}{p33} \naput[labelsep=1pt,npos=0.5]{$.4$}
\ncline[linewidth=0.75pt]{->}{c4}{p41} \naput[labelsep=1pt,npos=0.5]{$.4$}
\ncline[linewidth=0.75pt]{->}{c5}{p42} \nbput[labelsep=1pt,npos=0.5]{$.4$}
\ncline[linewidth=0.75pt]{->}{c5}{p43} \naput[labelsep=1pt,npos=0.5]{$.4$}
\ncline[linewidth=0.75pt]{->}{c5}{p51} \naput[labelsep=1pt,npos=0.3]{$.2$}
\ncline[linewidth=0.75pt]{->}{c6}{p51} \nbput[labelsep=1pt,npos=0.3]{$.2$}
\ncline[linewidth=0.75pt]{->}{c6}{p52} \naput[labelsep=1pt,npos=0.6]{$.4$}
\ncline[linewidth=0.75pt]{->}{c6}{p53} \naput[labelsep=1pt,npos=0.5]{$.4$}
\end{pspicture}
\begin{pspicture}(-1,-0.8)(8,8.5)
\pnode(3,7.25){A} \pnode(3,5.25){B} \ncline[linecolor=gray,linewidth=5pt,arrowsize=11pt]{->}{A}{B} \naput{reassignment}
\rput[r](-0.2,4.5){{\bf super-container:}}
\rput[r](-0.2,0.5){{\bf patterns:}}
\rput[c](1.3,4){ 
\drawRect{linewidth=0.75pt,fillstyle=solid,fillcolor=black!70!white}{0}{0}{0.25}{1}
\drawRect{linewidth=0.75pt,fillstyle=solid,fillcolor=black!50!white}{0.25}{0}{0.23}{1}
\drawRect{linewidth=0.75pt,fillstyle=solid,fillcolor=black!35!white}{0.48}{0}{0.21}{1}
\drawRect{linewidth=1.5pt,fillstyle=none}{0}{0}{0.69}{1}
\pnode(0.34,0){c1}
\rput[c](-0.2,0.5){$C_1$}
}
\rput[c](3.5,4){ 
\drawRect{linewidth=0.75pt,fillstyle=solid,fillcolor=black!25!white}{0}{0}{0.19}{1}
\drawRect{linewidth=0.75pt,fillstyle=solid,fillcolor=black!15!white}{0.19}{0}{0.17}{1}
\drawRect{linewidth=0.75pt,fillstyle=solid,fillcolor=black!15!white}{0.36}{0}{0.17}{1}
\drawRect{linewidth=1.5pt,fillstyle=none}{0}{0}{0.53}{1}
\pnode(0.26,0){c2}
\rput[c](-0.2,0.5){$C_2$}
}
\rput[c](0,0){
\drawRect{linewidth=0.75pt}{0}{0}{1}{1}
\drawRect{linewidth=0.75pt,fillstyle=solid,fillcolor=black!70!white}{0.0}{0}{0.75}{1}
\pnode(0.375,1){p1}
\rput[c](0.5,-1){$x_{p_1} = 0.4$}
}
\rput[c](1.3,0){
\drawRect{linewidth=0.75pt}{0}{0}{1}{1}
\drawRect{linewidth=0.75pt,fillstyle=solid,fillcolor=black!50!white}{0.0}{0}{0.69}{1}
\pnode(0.34,1){p2}
\rput[c](0.5,-1){$x_{p_2} = 0.4$}
}
\rput[c](2.6,0){
\drawRect{linewidth=0.75pt}{0}{0}{1}{1}
\drawRect{linewidth=0.75pt,fillstyle=solid,fillcolor=black!35!white}{0.0}{0}{0.63}{1}
\pnode(0.31,1){p3}
\rput[c](0.5,-1){$x_{p_3} = 0.4$}
}
\rput[c](3.9,0){
\drawRect{linewidth=0.75pt}{0}{0}{1}{1}
\drawRect{linewidth=0.75pt,fillstyle=solid,fillcolor=black!25!white}{0.0}{0}{0.57}{1}
\rput[c](0.5,-1){$x_{p_4} = 0.4$}
}
\rput[c](5.2,0){
\drawRect{linewidth=0.75pt}{0}{0}{1}{1}
\drawRect{linewidth=0.75pt,fillstyle=solid,fillcolor=black!15!white}{0.0}{0}{0.51}{1}
\rput[c](0.5,-1){$x_{p_5} = 0.4$}
}
\nccurve[angleA=-135,angleB=45,linewidth=1pt]{->}{c2}{p1} \nbput{$.4$}
\nccurve[angleA=-135,angleB=45,linewidth=1pt]{->}{c2}{p2} \naput[npos=0.6,labelsep=0pt]{$.4$}
\nccurve[angleA=-135,angleB=45,linewidth=1pt]{->}{c2}{p3} \naput{$.2$}
\end{pspicture}
\caption{Visualization of the reassignment in Lemma~\ref{lem:Gluing} for $k=3$. The upper packing graph is the red part of $G_2(x,y)$ with the optimal assignment $a$, assuming that each container has multiplicity $1$. 
The lower graph gives the red part of $G_2(\tilde{x},\tilde{y})$ with the constructed assignment $a$ that we give in the analysis. Darker colors indicate larger containers.}
\end{center}
\end{figure}

\begin{proof}
Consider the graph $G_2(x,y)$ as in section \ref{sec:Deficiency}. For every right node $(C,p)$, we assign $\text{mult}_\text{red}(C,p)= k\cdot \lfloor \frac{p_C}{k} \rfloor \cdot x_p$ for $C$ in size class $\sigma$, and $\text{mult}_\text{red}(C,p)=0$ for all other $C$.
We set $\text{mult}_\text{blue}(C,p)=\text{mult}(C,p)-\text{mult}_\text{red}(C,p)$.
By Lemma \ref{decomposition}, we can find integral red and blue multiplicities of left nodes so that $\text{def}(G_\text{red})=0$ and $\text{def}(G_\text{blue})\le \text{def}(x,y)+\sigma$. 
The red and blue graphs can now be treated separately, and so we restrict our attention to the red graph since it represents precisely the containers that we want to reassign.

For all nodes $(C,p)$ on the right of the red graph, we combine the copies of $C$ in pattern $p$ into containers of type $k\cdot C$.
For clarity we refer to these larger containers as super-containers. 
Similarly, we look at the containers of the left nodes, ordered from largest to smallest and taken with multiplicity.
In consecutive sets of cardinality $k$, we combine the containers into super-containers, except perhaps fewer than $k$ of the smallest ones. 
Write $C_i$ to represent the $i$th largest super-container on the left.

We claim that all super-containers except $C_1$ can be packed into the right nodes.
To see how to pack them, let $a$ be an optimal assignment in the original red graph.
For all $i$, $a$ assigned the containers making up $C_i$ to some combination of large-enough containers of total multiplicity $k$.
All such containers became part of super-containers in the new graph, and the total multiplicity of their contribution to these super-containers is exactly $1$.
These super-containers are not necessarily all large enough to fit $C_i$, but they are all large enough to fit $C_{i+1}$, and this is exactly where we send $C_{i+1}$.
With this assignment, at most one super-container and $k$ containers were left unpacked, and so the deficiency of the updated red graph is at most $2k\sigma$. 

For all containers $C$, we let $\tilde{y}_C=\text{mult}_\text{red}(C)+\text{mult}_\text{blue}(C)$. We note that we only changed $y$ by rearranging the containers, and in particular we did not change the item multiplicities. Therefore we know that 
$\text{def}(G_1(\tilde{y}))=\text{def}(G_1(y))$.
With this definition of $\tilde{y}$, we note that $G_{\text{red}}+G_{\text{blue}}$ is precisely $G_2(\tilde{x},\tilde{y}).$
We therefore have $\text{def}(G_2(\tilde{x},\tilde{y}))\le \text{def}(G_2(x,y))+O(k\sigma)$, and so
the total increase in deficiency is at most $O(k\sigma)$.
\end{proof}
We are now ready to give our second main lemma of this section.

\begin{lemma} [Reassigning containers]
Suppose $x\in \setR_{\ge 0}^{\pazocal{P}}, y\in \setZ_{\ge 0}^{\pazocal{C}}$, and $\sigma<2^{-4}$.
Then we can combine containers in size class $\sigma$ in $x$ and $y$ into larger containers, yielding new solutions $\tilde{x}, \tilde{y}$ satisfying the following conditions.
\begin{enumerate}
\item $\mathbf{1}^T\tilde{x}=\mathbf{1}^Tx$.
\item $|\text{supp}(\tilde{x})|\le |\text{supp}(x)|$.
\item For all patterns $p\in \text{supp}(\tilde{x})$ and containers $C$ in size class $\sigma$, $p_C\le (\frac{1}{\sigma})^{1/4}$.
\item Multiplicities of small containers in patterns in supp$(x)$ are not affected.
\item $\text{def}(\tilde{x},\tilde{y})\leq \text{def}(x,y)+O(\sigma^{3/4})$.
\end{enumerate}
\end{lemma}
\begin{proof}
We apply  Lemma~\ref{lem:Gluing} with parameter
 $k=\lfloor(\frac{1}{\sigma})^{1/4}\rfloor$ and obtain a pair $(\tilde{x},\tilde{y})$  so that $\text{def}(\tilde{x},\tilde{y})\le O(k\sigma)\le O(\sigma^{3/4})$, 
and so condition $(5)$ is satisfied. Since we have updated $x$ by altering the patterns in its support, conditions $(1)$ and $(2)$ are also satisfied. In the process of Lemma~\ref{lem:Gluing}, we decreased $p_C$ for $C$ in size class $\sigma$ to at most $k$. Since $\sigma<2^{-4}$, we know that $k\ge 2$, and so the containers we created are in strictly larger size classes. Therefore conditions $(3)$ and $(4)$ are satisfied. 
\end{proof}

Let us say briefly what the container reassignment does to the associated matrix $A$. If $\tilde{A}_C$ is any row of the updated matrix corresponding to a container in size class $\sigma$, we know $\tilde{A}_C$ is entrywise less than or equal to $A_C$ and $\|\tilde{A}_C\|_{\infty}\le (\frac{1}{\sigma})^{1/4}$. In all rows corresponding to smaller size classes, $\tilde{A}_C=A_C$.

Before we talk about applying Lovett-Meka, we want to summarize the results of our grouping and container reassignment.
We summarize the procedure:
\begin{enumerate}
\item[(1)] For size classes $s_{\min} \le \sigma \le 2^{-72}$, starting with the smallest, do:
  \begin{enumerate}
  \item[(2)] Group the containers in size class $\sigma$ with $\delta=\sqrt{\sigma}$. 
  \item[(3)] Whenever we find more than $(\frac{1}{\sigma})^{1/4}$ copies of the same container in one pattern, we put them together in a larger container.
  \end{enumerate}
\item[(4)] For $\sigma>2^{-72}$, group the containers in size class $\sigma$ with $\delta=64$.
\end{enumerate}
In the following we will call a size class $\sigma$ \emph{small} if $\sigma\le 2^{-72}$ and \emph{large} otherwise.
First note that the increase in deficiency of the entire procedure is at most\\
 \[
\sum_{\sigma\in2^{-\setN}}(O(\sigma^{1/2})+O(\sigma^{3/4}))+72 \cdot 64 =O(1).
\]

Let $A$ be the matrix we obtain at the end of this procedure. 
In addition, we would like to keep much of the group structure that was created during the procedure.
Define the \emph{shadow incidence matrix} $\tilde{A}$ to be the matrix that agrees with $A$ on large size classes, but for small size classes represents the incidences \emph{after step $(2)$}, but \emph{before step $(3)$}. 
We can imagine that whenever a container is put into a larger container, 
its incidence entry \emph{remains} in $\tilde{A}$. In particular a container might be put into containers
iteratively and hence it may contribute to several incidences in $\tilde{A}$ but only one in $A$. 
Note that $\tilde{A}$ is entrywise at least as large as $A$. 


For all containers $C\in \pazocal{C}$, let $A_C$ denote the row of $A$ corresponding to $C$, and $\tilde{A}_C$ the corresponding row of $\tilde{A}$. Recall that $A$ and $\tilde{A}$ contain columns for patterns in $\textrm{frac}(x)$.
Now, let us summarize the properties that
the container-forming procedure provides:
\begin{enumerate}
\item[(A)] For a container $C$ in size class $\sigma$ one has $\|\tilde{A}_C\|_1 \geq (\frac{1}{\sigma})^{1/2}$ if $\sigma$ is small, and $\|\tilde{A}_C\|_1=\|A_C\|_1 \geq 64$ if $\sigma$ is large.
\item[(B)] For a container $C$ in a small size class $\sigma$, and column $j=1,...,m,$ one has $A_{Cj} \leq (\frac{1}{\sigma})^{1/4}$.
\item[(C)] One has
\[
  \sum_{i=1}^s \|\tilde{A}_{C_i}\|_1 \cdot s(C_i)^{17/16} \leq 24\sum_{i=1}^s \|A_{C_i}\|_1 \cdot s(C_i).
\] 
\end{enumerate}
Here (A) follows from the fact that after step (2), we have $(\frac{1}{\sigma})^{1/2}$ incidences for each container. 
(B) follows since after step (3), there are at most $(\frac{1}{\sigma})^{1/4}$ containers of each type in a pattern. 
The condition in (C) can be understood as follows: if we have a container of size $s(C)$, then the containers in it may
appear many times in $\tilde{A}$ but only in smaller size classes. By discounting smaller incidences, we can upper-bound the contribution of 
the shadow incidences by the contribution of the actual containers. 

To make this more concrete, consider a container $C$ appearing in $A$ in some size class. If this container came from $k$ smaller containers, then those smaller containers are size at most $2\cdot \frac{s(C)}{k}$. Here the factor $2$ comes from the fact that during grouping our container could have been rounded down by a factor of $2$.
Therefore the contribution of the shadow incidences of these smaller containers to the left hand side is $(\frac{2s(C)}{k})^{17/16}\cdot k=s(C)^{17/16}\cdot 2^{17/16}k^{-1/16}.$ But we chose the parameters so that whenever we combine $k$ containers we have $k\ge 2^{18}$ and so the contribution is at most $2^{-1/16}\cdot s(C)^{17/16}$. The shadow incidences $\ell$ levels down similarly contribute $(2^{-1/16})^\ell\cdot s(C)^{17/16}$.
Then the total contribution of the shadows of $C$ to the left hand side of property (C) is at most
\[
  \sum_{\ell \geq 0} (2^{-1/16})^\ell s(C)^{17/16}\leq 24 \cdot s(C)^{17/16} \leq 24\cdot s(C).
\]

\section{Applying the Lovett-Meka algorithm\label{sec:ApplyingLM}}
Using the grouping and container reassignment above, we can replace $y$ with $\tilde{y}$ and $x$ with $\tilde{x}$ so that the incidence matrix $A$ and shadow matrix $\tilde{A}$ satisfy properties $(A)-(C)$. We now want to create intervals of the rows of $A$ and $\tilde{A}$ in a nice way so that we can apply Lovett-Meka and make $x$ more integral.
Formally, we will argue the following:
\begin{claim}
Suppose $x\in [0,1[^\pazocal{P},y\in \setZ^\pazocal{C}_{\geq 0}$, $A$ is the incidence matrix of $x$, and $\tilde{A}$ is a matrix so that $A$ and $\tilde{A}$ satisfy conditions $(A)+(B)+(C)$. Then there is a 
randomized polynomial time algorithm to find a vector $\tilde{x}$ satisfying
\begin{itemize*}
\item $\bm{1}^T\tilde{x} = \bm{1}^Tx$
\item $\textrm{def}(\tilde{x},y) \leq \textrm{def}(x,y) + O(1)$
\item $|\textrm{frac}(\tilde{x})| \leq \frac{1}{2} |\textrm{frac}(x)|$
\end{itemize*}
\end{claim}

Suppose the containers appearing in the patterns in supp$(x)$ are $C_1,...,C_s$, ordered from largest to smallest.
As we fix the fractional solution $x$ for now, let us denote $n(C_i) := \sum_{p \in \textrm{frac}(x)} p_{C_i} = \|A_{C_i}\|_1$
as the number of incidences of container $C_i$ in $A$. 
Similarly, let $\tilde{n}(i) = \|\tilde{A}_{C_i}\|_1$ be the number of incidences in the shadow matrix $\tilde{A}$. Again, we have $n(i) \leq \tilde{n}(i)$ for all $i$. 
Finally, let us denote $\tilde{n}_{\sigma} := \sum_{i\textrm{ in class }\sigma}\tilde{n}(i)$ as the total number of shadow incidences that occur for size class $\sigma$. 

For a fixed constant $K>0$, and for each small size class $\sigma$, we first create level $0$ intervals of the rows as follows. For any row $i$ satisfying $\tilde{n}(i)>\frac{1}{2}K(\frac{1}{\sigma})^{17/16}$, we let $\{i\}$ be its own interval. We then subdivide the remaining rows into intervals so that $\tilde{n}(I)\le K(\frac{1}{\sigma})^{17/16}$ for each interval $I$. We need a total of at most $\frac{4}{K}\sigma^{17/16}\tilde{n}_\sigma+1$ intervals on level $0$.

Now, given an interval $I$ on level $\ell$ with $|I|>1$, we will subdivide $I$ into at most $3$ intervals on level $\ell+1$.
First, for any row $i \in I$ with $\tilde{n}(i)>(\frac{1}{2})^{\ell+1}K({\frac{1}{\sigma}})^{17/16}$, let $\{i\}$ be its own interval. We then subdivide the remaining rows into intervals so that $\tilde{n}(I)\le (\frac{1}{2})^{\ell}K(\frac{1}{\sigma})^{17/16}$. Since none of the rows $i\in I$ became its own interval on level $\ell$, we also know that $\tilde{n}(i)\le (\frac{1}{2})^{\ell}K(\frac{1}{\sigma})^{17/16}$, and so in fact this bound holds for every interval on level $\ell+1$. 
The number of intervals on level $\ell$ is at most $3^\ell\cdot (\frac{4}{K}\sigma^{17/16}\tilde{n}_\sigma+1)$.

For large size classes $\sigma$, 
create an interval for each row $\{i\}$. Due to the grouping procedure, the size of each interval is at least $64$. All such intervals are level zero, and we do not create any higher levels.

Let us abbreviate all intervals on level $\ell$ for size class $\sigma$ as $\pazocal{I}_{\sigma,\ell}$. 
We denote $\pazocal{I}_{\sigma} := \bigcup_{\ell \geq 0} \pazocal{I}_{\sigma,\ell}$ as the whole family for size class $\sigma$
and $\pazocal{I} := \bigcup_{\sigma} \pazocal{I}_{\sigma}$ 
as the union over all size classes. 

For an interval $I$, we define the vector
\[
  v_I := \sum_{i \in I} A_i
\]
as the sum of the corresponding rows in the incidence matrix.

For an interval $I \in \pazocal{I}_{\sigma,\ell}$, we define $\lambda_I := \ell$ (that means the parameter just denotes the level
on which it lives).
The input for the Lovett-Meka algorithm will consist of the pairs $\{ (v_I,\lambda_I)\}_{I \in \pazocal{I}}$ where we use $\lambda_I \geq 0$ as the parameter for a constraint 
with normal vector $v_I$. Additionally, we add a single vector $v_{\textrm{obj}} := \bm{1}$ with parameter $\lambda_{\textrm{obj}} := 0$ to
control the objective function. 
There are two things to show. First we argue that the parameters are chosen so that the condition of the Lovett-Meka algorithm
is actually satisfied:
\begin{lemma}
Suppose that $|\text{supp}(x)|\ge L\log(\frac{1}{s_{\min}})$. For $K,L$ large enough constants, one has 
\[
 \sum_{I \in \pazocal{I}} e^{-\lambda_I^2/16} +1 \le \frac{1}{16} \cdot |\textrm{supp}(x)|
\]
\end{lemma}
\begin{proof}
On level $0$, we have $|\pazocal{I}_{\sigma,0}| \le \frac{4}{K}\sigma^{17/16} \cdot \tilde{n}_{\sigma}+1$ many intervals
and hence on level $\ell \geq 0$ there are $|\pazocal{I}_{\sigma,\ell}| \le 3^{\ell} \cdot (\frac{4}{K}\sigma^{17/16}\cdot \tilde{n}_{\sigma}+1)$ many.
We can calculate that
\begin{eqnarray*}
 \sum_{I \in \pazocal{I}} e^{-\lambda_I^2/16} 
 &=& 
 \sum_{\sigma \text{ small}} \sum_{\ell \geq 0} e^{-\ell^2/16} \cdot |\pazocal{I}_{\sigma,\ell}| +
  \sum_{\sigma \text{ large}} \sum_{\ell \geq 0} e^{-\ell^2/16} \cdot |\pazocal{I}_{\sigma,\ell}| \\
&\le& 
\sum_{\sigma \text{ small}} \sum_{\ell \geq 0} e^{-\ell^2/16} \cdot 3^{\ell} \cdot (\frac{4}{K}\sigma^{17/16} \cdot \tilde{n}_{\sigma}+1) 
+ \sum_{\sigma \text{ large}} |\pazocal{I}_{\sigma,0}|\\
 &\stackrel{\textrm{for } K,L \textrm{ large enough}}{\leq}& 
 \frac{1}{64}|\text{supp}(x)|+\frac{1}{128 \cdot 24}\sum_{\sigma \text{ small}} \sigma^{17/16} \cdot \tilde{n}_{\sigma} 
 +\frac{1}{64}|\text{supp}(x)|\\
 & \stackrel{\textrm{property }(C)}{\leq}& 
 \frac{1}{64}|\text{supp}(x)|+\frac{1}{128}\sum_{\sigma} \sigma \cdot n_{\sigma} 
  +\frac{1}{64}|\text{supp}(x)|\\
&\leq&
\frac{3}{64}\cdot |\text{supp}(x)| \le
 \frac{1}{16} \cdot |\textrm{supp}(x)|-1.
\end{eqnarray*}
We used that the total size for each pattern is at most $1$, and so the sum of the sizes of all incidences in the matrix $A$ is at most $|\text{supp}(x)|$.
\end{proof}

Now, suppose we do run the Lovett-Meka algorithm and obtain a solution $\tilde{x}$ with $|\text{frac}(\tilde{x})| \leq \frac{1}{2}|\textrm{frac}(x)|$ so that 
\[
  |\left<v_I,x-\tilde{x}\right>| \leq \lambda_I \cdot \|v_I\|_2 \quad \forall I \in \pazocal{I} \quad \textrm{and} \quad \bm{1}^Tx = \bm{1}^T\tilde{x}.
\]
The following is crucial to our error analysis: the lengths $\|v_I\|_2$ that appear in the error bound are not 
too long and in particular the ratio $\frac{\|v_I\|_2}{\tilde{n}(I)}$ decreases with smaller container sizes. 
\begin{lemma}
Fix an interval $I \in \pazocal{I}_{\sigma,\ell}$ where $\sigma$ is small. Then $\|v_I\|_2 \leq \tilde{n}(I) \cdot \sigma^{1/8}$.
\end{lemma}
\begin{proof}
Recall that $v_I = \sum_{i \in I} A_{C_i}$ where each row $A_{C_i}$ has a row-sum of $\|A_{C_i}\|_1 \leq \|\tilde{A}_{C_i}\|_1$. We have $\tilde{n}(i)= \|\tilde{A}_{C_i}\|_1\geq (\frac{1}{\sigma})^{1/2}$, while $\|A_{C_i}\|_{\infty} \leq (\frac{1}{\sigma})^{1/4}$. Therefore, we have
\[
  \|A_{C_i}\|_2 \leq \sqrt{\|A_{C_i}\|_1 \cdot \|A_{C_i}\|_{\infty}} \leq \sqrt{\|\tilde{A}_{C_i}\|_1 \cdot \|A_{C_i}\|_{\infty}} = \|\tilde{A}_{C_i}\|_1\sqrt{\frac{\|A_{C_i}\|_{\infty}}{\|\tilde{A}_{C_i}\|_1}} \leq  \tilde{n}(i) \cdot \sigma^{1/8}.
\]
Then by the triangle inequality $\|v_I\|_2 \leq \sum_{i \in I} \|A_{C_i}\|_2 \leq \tilde{n}(I) \cdot \sigma^{1/8}$.
\end{proof}
The next step should be to argue that the error in terms of the deficiency will be small. Recall that we still assume
that containers are sorted so that $1 \geq s(C_1) \geq s(C_2) \geq \ldots \geq s(C_s)>0$.
\begin{lemma}\label{lem:LMRounding}
Let $C_i$ be a container in small size class $\sigma$. Then
\[
   \Big|\sum_{j \leq i} A_{C_j}(x-\tilde{x}) \Big| \leq O\left(\frac{1}{\sigma}\right)^{15/16}.
\]
If $C_i$ is a large container, then $\sum_{j\leq i}A_{C_j}(x-\tilde{x})=0$.
\end{lemma}
\begin{proof}
If $C_i$ is a container in small size class $\sigma$, we can write the interval $\{ 1,\ldots,i\} = \dot{\bigcup}_{I \in \pazocal{I}(i)} I$ as the disjoint union of intervals $\pazocal{I}(i) \subseteq \pazocal{I}$ from our collection 
so that the only intervals $I \in \pazocal{I}(i)$ with $\lambda_I > 0$ that we are using are from class $\sigma$ and we only take at most three
intervals from each level; for all three such intervals on level $\ell$, we have $\|v_I\|_2 \le \tilde{n}(I)\sigma^{1/8}\le K\cdot2^{-\ell}\left(\frac{1}{\sigma}\right)^{15/16}$.
Consequently, we can bound
\begin{eqnarray*}
  \Big|\sum_{j \leq i} A_{C_j}(\tilde{x}-x)\Big| &\leq& \sum_{I \in \pazocal{I}(i)} \lambda_I \cdot \|v_I\|_2 \leq  \sum_{\ell \geq 0} 3\ell \cdot K\cdot2^{-\ell}\left(\frac{1}{\sigma}\right)^{15/16}\\
  &= & O(1) \cdot \left(\frac{1}{\sigma}\right)^{15/16}.
\end{eqnarray*}
If $C_i$ is a large container, we can write $\{1,...,i\}$ as a disjoint union of intervals with $\lambda=0$, and so the statement holds.
\end{proof}

It remains to argue why $\textrm{def}(\tilde{x},y) \leq \textrm{def}(x,y) + O(1)$ for one application of Lovett-Meka. 
First notice that $A_{C_j}\tilde{x}=\text{mult}(C_j,\tilde{x})$ and $A_{C_j}x=\text{mult}(C_j,x)$. Therefore by Lemmas \ref{lem:DefIncrease} and \ref{lem:LMRounding} the rounding of each size class $\sigma$ increases the deficiency by at most $O(1)\cdot(\frac{1}{\sigma})^{15/16}\cdot \sigma=O(1)\cdot \sigma^{1/16}$. Summing over all size classes gives a total increase in deficiency 
\[
O(1)\cdot \sum_{\sigma \in 2^{-\setN}} \sigma^{1/16} \leq O(1).
\]


\bibliographystyle{alpha}
\bibliography{binpacking-log-apx}

\appendix

\section*{Appendix A}

Here we give the proof of Lemma~\ref{lem:AssumptionSizesAtLeast1-n}. 
Recall that for our result we would need $f(k) = \Theta(\log(k))$.
\begin{proof}
Let $(s,b)$   be any bin packing instance with $s \in [0,1]^n$ and $b \in \setN^n$. 
Let $U := \sum_{i=1}^n b_is_i$ be the total size. Note that $U \leq OPT_f \leq 2U+1$, 
so $U$ is a good estimate on the value of the LP optimum. 
We split items into \emph{large} ones $L := \{ i \in [n] \mid s_i \geq \frac{1}{U}\}$
and \emph{small} ones $S := \{ i \in [n] \mid s_i < \frac{1}{U}\}$.

Now, we perform the \emph{geometric grouping} from~\cite{KarmarkarKarp82}
to the large items as follows: sort items consecutively and form groups of
total size between $2$ and $3$. Then for each group, round all items to the largest item 
type in its group. 
This procedure allows to reduce the number of different item 
types to $U$ while the optimal fractional value increases to at most $OPT_f' \leq OPT_f + O(\log U)$.
Now we run the assumed algorithm to assign items in $L$ to at most $OPT_f' + f(U) \leq OPT_f + f(OPT_f) + O(\log U)$
bins. Here we are using that $OPT_f \geq U$ is an upper bound on the number of items in the
modified instance and $s_{\min} := \frac{1}{U}$ is a lower bound on the item sizes in $L$. 

Then we ``sprinkle'' the small items greedily over those bins. 
If no new bin needs to be opened, we are done. 
Otherwise, we know that the solution consists of $k$ bins such that  $k-1$ bins are at least $1 - \frac{1}{U}$ full. This implies 
$U \geq (k-1) \cdot (1 - \frac{1}{U})$, and hence $k \leq U + 3 \leq OPT_f + 3$ assuming $U \geq 2$.
\end{proof}

\end{document}